\PassOptionsToPackage{table,xcdraw}{xcolor}	%
\documentclass[acmsmall, screen, anonymous=false]{acmart}

\usepackage[utf8]{inputenc}

\usepackage{xspace}

\usepackage{subcaption} %

\newtheorem{remark}{Remark}

\newtheorem{hypothesis}{Hypothesis}

\def\polylog{\operatorname{polylog}} %

\newcommand{\sparseBMM}{\textsc{sparseBMM}}
\newcommand{\hyperclique}{\textsc{Hyperclique}}
\newcommand{\ThreeSUM}{\textsc{3sum}}
\newcommand{\seth}{\textsc{Seth}}
\newcommand{\calV}{{\mathcal V}}

\newcommand{\calH}{{\mathcal H}}
\newcommand{\calS}{{\mathcal S}}
\newcommand{\atoms}{\texttt{atoms}}
\newcommand{\var}{\texttt{var}}
\newcommand{\free}{\texttt{free}}
\newcommand{\bigO}{{\mathcal{O}}} %
\newcommand{\mh}{\textrm{mh}}
\newcommand{\mhfree}{\textrm{fmh}}
\newcommand{\RF}{\textrm{RF}}
\newcommand{\freeind}{\alpha_{\textrm{free}}}

\newcommand{\dom}{\texttt{dom}}
\newcommand{\ar}{\texttt{ar}}
\newcommand{\lex}{L}
\newcommand{\vertex}{V}
\newcommand{\varidx}{p}

\def \factor{\textrm{factor}}
\def \bucket{\textrm{bucket}}
\def \root{\textrm{root}}
\def \tuple{t}
\def \w{\texttt{weight}}
\def \layer{\texttt{layer}}
\def \rng{\texttt{start}}
\def \endIndex{\texttt{end}}

\newcommand{\N}{\mathbb{N}} %
\newcommand{\R}{{\mathbb{R}}} %

\newcommand{\Comp}[2]{\ensuremath{\langle #1, #2 \rangle}}

\def\e#1{\emph{#1}}

\def\angs#1{\mathord{\langle #1\rangle}}
\AtBeginDocument{%
  \providecommand\BibTeX{{%
    \normalfont B\kern-0.5em{\scshape i\kern-0.25em b}\kern-0.8em\TeX}}}

\newcommand{\smallsection}[1]{\vspace{5mm}\noindent\textbf{#1.}} %
\newcommand{\introparagraph}[1]{\textbf{#1.}} %
\newcommand{\datarule}{{\,:\!\!-\,}} %
\newcommand{\nikos}[1]{{{\color{blue}{[\textbf{Nikos}: #1]}}}} %
\newcommand{\nofar}[1]{{{\color{violet}{[\textbf{Nofar}: #1]}}}} %
\newcommand{\benny}[1]{{{\color{olive}{[\textbf{Benny}: #1]}}}} %
\newcommand{\wolf}[1]{{{\color{magenta}{[\textbf{WG}: #1]}}}} %
\newcommand{\mirek}[1]{{{\color{cyan}{[\textbf{MR}: #1]}}}} %
\newcommand\reviewer[2]{{\color{red}[\textbf{Reviewer #1:} #2]}}

\newcommand{\nikosrem}[1]{{{\color{blue}{[\textbf{Reminder}: #1]}}}}
\newcommand{\nofarrem}[1]{{{\color{blue}{[\textbf{Reminder}: #1]}}}}
\newcommand{\wolfrem}[1]{{{\color{blue}{[\textbf{Reminder}: #1]}}}}
\newcommand{\mirekrem}[1]{{{\color{blue}{[\textbf{Reminder}: #1]}}}}

\usepackage{tabularx} %
\usepackage{booktabs} %
   
\newcommand{\ourversion}[1]{#1} 
\newcommand{\hide}[1]{} 
\newcommand{\hidetwo}[2]{}

\newcommand{\mkclean}{

    \renewcommand{\ourversion}{\hide}

    \renewcommand{\mirek}{\hide}
    \renewcommand{\wolf}{\hide}
    \renewcommand{\nikos}{\hide}
    \renewcommand{\nofar}{\hide}
    \renewcommand{\benny}{\hide}
    \renewcommand{\reviewer}{\hidetwo}

    \renewcommand{\nikosrem}[1]{} 
    \renewcommand{\nofarrem}[1]{}
    \renewcommand{\wolfrem}[1]{} %
    \renewcommand{\mirekrem}[1]{}
}

\mkclean %

\usepackage[ruled,noend,linesnumbered]{algorithm2e} %

\usepackage{setspace} %

\DontPrintSemicolon     %
\SetNlSty{}{}{}                %
\SetAlgoInsideSkip{smallskip}   %

\SetAlFnt{\small}			%
\SetAlCapFnt{\small}		%
\SetAlCapNameFnt{\small}

\usepackage[capitalise,nameinlink]{cleveref} %
\crefformat{footnote}{#2\footnotemark[#1]#3}

\AtEndPreamble{%
    \hypersetup{colorlinks,
      linkcolor=purple,
      citecolor=blue,
      urlcolor=ACMDarkBlue,
      filecolor=ACMDarkBlue}}

\def\rel#1{\mathsf{#1}}
\def\att#1{\mathrm{#1}}
\def\val#1{\mathtt{#1}}
\def\set#1{\mathord{\{#1\}}}

\newcommand{\eat}[1]{}

\usepackage{tcolorbox}		%
\tcbuselibrary{breakable,skins}		%

\tcbset{examplestyle/.style={
		enhanced jigsaw,	%
		colback=gray!10,	%
		colframe=gray!10,	%
		arc=2mm,			%
		boxrule=0pt,		%
		left=1mm,
		right=1mm,
		left skip=0mm,  %
		right skip=0mm, %
		top=1mm,		%
		bottom=1mm,		%
		breakable,		%
		parbox = false,		%
		before={\par\pagebreak[0]\vspace{2mm}\parindent=0pt},
		after={\par\pagebreak[0]\vspace{1mm}\parindent=0pt},				
		bottomrule = 0mm,
		boxsep = 0mm,					%
		topsep at break=0pt,			%
		bottomsep at break=0pt,			%
		pad at break=0mm,
		pad before break=0mm,		
		pad after break=1mm,		
		bottomrule at break=0mm,
		toprule at break=0mm,		
		}}

\tcolorboxenvironment{example}{examplestyle}

\setcopyright{none}
\begin{document}
\title[Tractable Orders for Direct Access to Ranked Answers of CQs]{Tractable Orders for Direct Access to Ranked Answers of Conjunctive Queries}
\author{Nofar	Carmeli}
\authornote{Both authors contributed equally to the paper.}
\orcid{0000-0003-0673-5510}
\email{Nofar.Carmeli@inria.fr}
\affiliation{Technion\country{Israel}}
\affiliation{DI ENS, ENS, CNRS, PSL University, Inria\country{France}}

\author{Nikolaos	Tziavelis}
\authornotemark[1]
\orcid{0000-0001-8342-2177}
\email{tziavelis.n@northeastern.edu}
\affiliation{Northeastern University\country{USA}}

\author{Wolfgang	Gatterbauer}
\orcid{0000-0002-9614-0504}
\email{w.gatterbauer@northeastern.edu}
\affiliation{Northeastern University\country{USA}}

\author{Benny Kimelfeld}
\orcid{0000-0002-7156-1572}
\email{bennyk@cs.technion.ac.il}
\affiliation{Technion - Israel Institute of Technology\country{Israel}}

\author{Mirek	Riedewald}
\orcid{0000-0002-6102-7472}
\email{m.riedewald@northeastern.edu}
\affiliation{Northeastern University\country{USA}}

\begin{abstract}
We study the question of when we can provide \emph{direct access to the $k$-th answer} to a Conjunctive Query (CQ) according to a specified order over the answers in time logarithmic in the size of the database, following a preprocessing step that constructs a data structure in time quasilinear in database size.
Specifically, we embark on the challenge of identifying \emph{the tractable answer orderings}, that is, those orders that allow for such complexity guarantees.  To better understand the computational challenge at hand, we also investigate the more modest task of providing access to only a single answer (i.e., finding the answer at a given position), a task that we refer to as \emph{the selection problem}, and ask when it can be performed in quasilinear time. We also explore the question of when selection is indeed easier than ranked direct access.
We begin with \emph{lexicographic orders}. For each of the two problems, we give a decidable characterization (under conventional complexity assumptions) of the class of tractable lexicographic orders for every CQ without self-joins.  We then continue to the more general \emph{orders by the sum of attribute weights} and establish the corresponding decidable characterizations, for each of the two problems, of the tractable CQs without self-joins. Finally, we explore the question of when the satisfaction of Functional Dependencies (FDs) can be utilized for tractability, and establish the corresponding generalizations of our characterizations for every set of unary FDs.

\end{abstract}

\maketitle

\section{Introduction}
\label{sec:intro}

\emph{When can we support direct access
to a ranked list of answers to
a database query without (and
considerably faster than) materializing all
answers?}
To illustrate the concrete instantiation of this question, assume the following simple relational schema 
for
information about activities of residents in the context of pandemic spread:
\[
  \rel{Visits}(\att{person},\att{age},\att{city})
  \quad\quad 
\rel{Cases}(\att{city},\att{date},\att{\#cases})
\]
Here, the relation $\rel{Visits}$ mentions, for each person in the database, the cities that the
person visits regularly (e.g., for work and for visiting relatives) and the age of
the person (for risk assessment); 
the relation $\rel{Cases}$ specifies
the number of new infection cases in specific cities at specific dates
(a measure that is commonly used for spread assessment albeit being
sensitive to the amount of testing).

Suppose that we wish to
efficiently compute the natural join $\rel{Visits}\Join \rel{Cases}$
based on equality of the city attribute, so that we have 
all combinations of people (with their age), the cities they regularly
visit, and the city's daily new cases.  For
example,\[(\val{Anna},\val{72},\val{Boston},\val{12/7/2020},\val{179})\,.\]
While the number of such answers could be quadratic in the size of the
database, the seminal work of Bagan, Durand, and
Grandjean~\cite{bdg:dichotomy} has established that the answers can be
enumerated with a constant delay between
consecutive answers, after a linear-time preprocessing phase.  This is
due to the fact that this join is a special case of a \e{free-connex}
Conjunctive Query (CQ).
In the case of CQs without self-joins, being free-connex 
is a sufficient and necessary condition for such efficient
evaluation~\cite{bb:thesis,bdg:dichotomy}.  The necessity requires
conventional assumptions in fine-grained complexity\footnote{For the
  sake of simplicity, throughout this section we make all of these
  complexity assumptions.  We give their formal statements in
  \Cref{sec:hypotheses}.}
and it holds even if we multiply the preprocessing time and delay by a
logarithmic factor
in the size of the database.\footnote{We refer to
  those as \emph{quasilinear preprocessing} and \emph{logarithmic delay},
  respectively.}
To realize the constant (or logarithmic) delay, 
the preprocessing
phase builds a data structure that enables efficient iteration over
the answers in the enumeration phase. Brault-Baron~\cite{bb:thesis}
showed that in the linear-time preprocessing phase, we can build a
structure with better guarantees: not only log-delay enumeration, but
even log-time \e{direct access}: 
a structure that, given $k$,
allows to directly retrieve the $k^\textrm{th}$ answer in the
enumeration without needing to enumerate the preceding
$k-1$ answers first.\footnote{``Direct access'' is also widely known as ``random access.'' We prefer to use ``direct access'' to avoid confusion with the problem of answering ``in random order.'' }
Later, Carmeli et al.~\cite{oldRandomAccess} 
showed how such a structure can be used for enumerating answers in
a random order (random permutation)\footnote{Not to be confused with ``random access.''}
with the statistical guarantee that the order is uniformly distributed.  In particular, in the above example we can
enumerate the answers of $\rel{Visits}\Join \rel{Cases}$ in a provably
uniform random permutation (hence, ensuring statistical validity of
each prefix) with logarithmic delay, after a linear-time preprocessing
phase.  Their direct-access structure also allows for
\e{inverted access}: given an answer, return the index $k$ of that
answer (or determine that it is not a valid answer).
Recently, Keppeler~\cite{Keppeler2020Answering} proposed another direct-access structure with the additional ability to allow efficient database updates, but at the cost of only supporting a limited subset of free-connex CQs.
All known direct-access structures~\cite{bb:thesis,oldRandomAccess,Keppeler2020Answering}
allow the answers to be \e{sorted} by some lexicographic order (even
if the formal results do not explicitly state it).
For instance, in our example of $\rel{Visits}\Join \rel{Cases}$, 
the structure could be such that the tuples are enumerated in the (descending or ascending) 
order of $\att{\#cases}$ and then by date, or in the order of $\att{city}$ and then by $\att{age}$.
 Hence, in logarithmic time we can evaluate quantile queries, namely find the $k^\textrm{th}$ answer in order, and determine the position of a tuple inside the sorted list. From this we can also
conclude (fairly easily) 
 that we can enumerate the answers ordered by $\att{age}$ where ties are broken randomly, again provably uniformly.  
 Carmeli et al.~\cite{oldRandomAccess} have also shown how the order of the answers can be useful for generalizing 
 direct-access algorithms from CQs to UCQs.
Notice that direct access to the sorted list of answers is a stronger requirement than \e{ranked enumeration} that has been studied in recent work~\cite{tziavelis20vldb,tziavelis20tutorial,deep19,YangRLG18:anyKexploreDB,DBLP:conf/icdt/BourhisGJR21,DBLP:conf/www/YangAGNRS18}.

Yet, the choice of lexicographic order is an artefact
of the structure construction, e.g., the elimination
order~\cite{bb:thesis}, the join tree~\cite{oldRandomAccess}, or the $q$-tree \cite{Berkholz17updates}. 
If
the application desires a
specific lexicographic order, we can only
hope to find a matching construction. 
However, this is not necessarily
possible.  For example, could we construct in (quasi)linear time a
direct-access structure for $\rel{Visits}\Join \rel{Cases}$ ordered by
$\att{\#cases}$ and then by $\att{age}$? 
Interestingly, 
we will show that
the answer is negative: it is impossible to build in quasilinear time
a direct-access structure with logarithmic access time
for that lexicographic order.

Getting back to the question posed at the beginning of this section,
in this paper we embark on the challenge of identifying, for each CQ,
the orders that allow for efficiently constructing a direct-access
structure. We adopt the \emph{tractability yardstick} of quasilinear
construction (preprocessing) time and logarithmic access time. In
addition, we focus on two types of orders: lexicographic orders, and
scoring by the sum of attribute weights.

As aforesaid, some of the orders that we study are intractable. To
understand the root cause of the hardness, we consider another task
that allows us to narrow our question to a considerably weaker
guarantee.  Our notion of tractability so far requires the
construction of a structure in quasilinear time and a direct access in
logarithmic time.  In particular, if our goal is to compute just a
\emph{single} quantile, say the $k^\textrm{th}$ answer, then it takes
quasilinear time.  Computing a single quantile is known as the
\e{selection problem}~\cite{blum73select}.  The question we ask is to
what extent selection is a weaker requirement than direct access in
the case of CQs. In other words, how much larger is the class of
ordered CQs with quasilinear selection than that of CQs with a
quasilinear construction of a logarithmic-access structure?

In some situations, we might be able to avoid hardness through a more
careful inspection of the \e{integrity constraints} that the database
guarantees on the source relations.  For example, it turns out that we
\e{can} construct in quasilinear time a direct-access structure for
$\rel{Visits}\Join \rel{Cases}$ ordered by $\att{\#cases}$ and then by
$\att{age}$ if we assume that each city occurs at most once in
$\rel{Cases}$ (i.e., for each city we have a report for a single
day). Hence, it may be the case that an ordered CQ is classified as
intractable (with respect to the guarantees that we seek), but it
becomes tractable if we are allowed to assume that the input database
satisfies some integrity constraints such as key constraints or more
general Functional Dependencies (FDs). Moreover, FDs are so common
that ignoring them implies that we often miss opportunities of fast
algorithms.  This angle arises regardless of answer ordering, and
indeed, Carmeli and Kr{\"{o}}ll~\cite{DBLP:journals/mst/CarmeliK20}
showed precisely how the class of (self-join-free) CQs with tractable enumeration
extends in the presence of FDs. Accordingly, we extend our study on
ranked direct access and the selection problem to incorporate FDs, and
aim to classify every combination of \e{(a)}  CQ, \e{(b)} order over
the answers, \e{and (c)}  set of FDs.

\begin{figure*}[t]
\centering
\includegraphics[width=\linewidth]{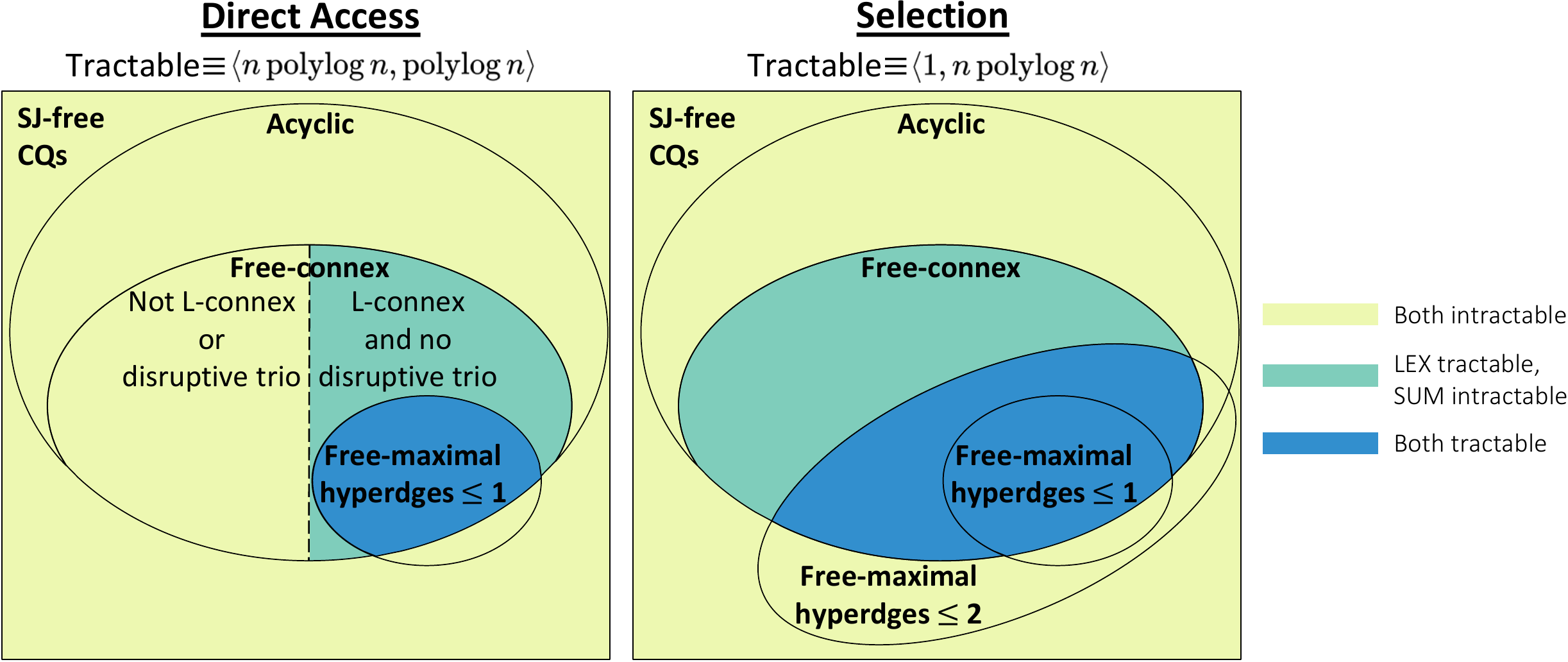}
\caption{Overview of our results for lexicographic (LEX) orders 
and sum-of-weights (SUM) orders. 
CQs without self-joins (SJ-free) are classified based on the  tractability of the \emph{direct access problem} (left) and the \emph{selection problem} (right).
The $L$-connex property applies only to lexicographic orders $L$ (the precise definitions are given in \Cref{sec:prelim}).
All tractable cases extend to CQs with self-joins. 
The sizes of the ellipses are arbitrary and do not correspond to the size or importance of the classes.}
\label{fig:overview}
\end{figure*}

\paragraph{Contributions} Before we describe the results that we
establish in this manuscript, let us illustrate them on an example.

\begin{figure}[t]
\centering
\begin{subfigure}{.23\linewidth}
    \centering
	\setlength{\tabcolsep}{1mm}
	\hspace{0mm}
	\mbox{
		\begin{tabular}[t]{ >{$}c<{$} | >{$}c<{$}  >{$}c<{$} }
		\mathbf{R}	&  x & y \\
		\hline
			& 1 & 5 		\\			
			& 1 & 2		\\
			& 6 & 2  		\\
		\multicolumn{3}{c}{}	
		\end{tabular}			
	}
	\hspace{0mm}
	\mbox{
		\begin{tabular}[t]{ >{$}c<{$} | >{$}c<{$} >{$}c<{$} }
		\mathbf{S}	& y		& z	 \\
		\hline
			& 5	& 3		\\
			& 5	& 4		\\
			& 5	& 6		\\			
			& 2	& 5		\\
		\multicolumn{3}{c}{}		
		\end{tabular}
	}
	\caption{Example Database.}
\end{subfigure}	
\hfill
\begin{subfigure}{.24\linewidth}
    \centering
	\setlength{\tabcolsep}{1mm}
	\hspace{0mm}
	\mbox{
		\begin{tabular}[t]{ >{$}c<{$} | >{$}c<{$}  >{$}c<{$} >{$}c<{$} }
		\mathbf{Q}	&  x 	& y   & z \\
		\hline
		\#1	& 1 & 2 & 5 	\\			
		\#2	& 1 & 5 & 3 	\\
		\rowcolor{blue!12}
		\#3	& 1 & 5 & 4		\\ 					
		\#4	& 1 & 5 & 6		\\ 				
		\#5	& 6 & 2 & 5		\\			
		\end{tabular}			
	}
	\label{fig:intro_lex1}
	\caption{LEX ordering $\angs{x, y, z}$.}
\end{subfigure}	
\begin{subfigure}{.24\linewidth}
    \centering
	\setlength{\tabcolsep}{1mm}
	\hspace{0mm}
	\mbox{
		\begin{tabular}[t]{ >{$}c<{$} | >{$}c<{$}  >{$}c<{$} >{$}c<{$} }
		\mathbf{Q}	&  x 	& z   & y \\
		\hline
		\#1	& 1 & 3 & 5 	\\			
		\#2	& 1 & 4 & 5 	\\
		\rowcolor{blue!12}		
		\#3	& 1 & 5 & 2		\\ 	
		\#4	& 1 & 6 & 5 	\\						
		\#5	& 6 & 5 & 2		\\			
		\end{tabular}			
	}
	\label{fig:intro_lex2}
	\caption{LEX ordering $\angs{x, z, y}$}
\end{subfigure}	
\begin{subfigure}{.26\linewidth}
    \centering
	\setlength{\tabcolsep}{1mm}
	\hspace{0mm}
	\mbox{
		\begin{tabular}[t]{ >{$}c<{$} | >{$}c<{$}  >{$}c<{$} >{$}c<{$} | >{$}c<{$}}
		\mathbf{Q}	&  x 	& y   & z	& x + y + z \\
		\hline
		\#1	& 1 & 2 & 5 & 8	\\
		\#2	& 1 & 5 & 3 & 9	\\
		\rowcolor{blue!12}		
		\#3	& 1 & 2 & 6 & 9	\\	
		\#4	& 1 & 5 & 4	& 10	\\ 	 
		\#5	& 6 & 2 & 5	& 13	\\							
		\end{tabular}			
	}
	\label{fig:intro_sum}
	\caption{SUM ordering.}
\end{subfigure}	
\caption{\Cref{ex:intro}: An example input database (a) and possible orderings of the answers to the query $Q(x, y, z) \datarule R(x, y), S(y, z)$.
The orderings in (b) and (c) use two different lexicographic orders (LEX),
while the ordering in (d) uses a sum-of-weights order where the weights are assumed to be identical to the attribute values.}
\label{fig:intro}
\end{figure}

\begin{example}
\label{ex:intro}
\Cref{fig:intro} depicts an example database and different orderings of the answers 
to the 2-path CQ $Q(x, y, z) \datarule R(x, y), S(y, z)$.
The question we ask is whether the median 
(e.g., the 3rd answer in the example) or in general, the answer in any index
can be computed efficiently as the database size $n$ grows.
Tractable direct access requires $\bigO(\polylog n)$ per access after $\bigO(n \polylog n)$ preprocessing,
while tractable selection requires $\bigO(n \polylog n)$ for a single access.
We compare the impact of different orders, projections, and FDs on the example 2-path CQ.
\begin{itemize}
    \item LEX $\angs{x, y, z}$: Direct access is tractable.
    \item LEX $\angs{x, z, y}$: Direct access is intractable because
      the lexicographic order ``does not agree'' with the query
      structure, which we capture through the concept of a
      \e{disruptive trio} that we introduce later on. However, selection is tractable.
    \item LEX $\angs{x, z}$: Direct access is intractable because the query is not $L$-connex for the partial lexicographic order $L$. However, selection is again tractable.
    \item LEX $\angs{x, z}$ and $y$ projected away: Selection is now intractable because the query is not free-connex.
    \item LEX $\angs{x, z, y}$, with the FD $R: y \rightarrow x$ or the FD $S: y \rightarrow z$: Direct access is tractable as a consequence of earlier work on enumeration with FDs~\cite{DBLP:journals/mst/CarmeliK20}.
    \item LEX $\angs{x, z, y}$, with the FD $R: x \rightarrow y$: Direct access is tractable with the techniques that we develop in this paper. Intuitively, the FD implies that the order is equivalent to the tractable order $\angs{x, y, z}$.
    \item LEX $\angs{x, z, y}$, with the FD $S: z \rightarrow y$: Direct access is intractable since the FD does not help in this case.
    \item SUM $x+y+z$: Direct access is intractable because it would allow us to solve the 3SUM problem in subquadratic time, yet selection is tractable.
    \item SUM $x+y$ and $z$ projected away: Direct access is tractable because all the free variables are contained in $R$.
    \item SUM $x+z$ and $y$ projected away: Selection is intractable because the query is not free-connex.    
\end{itemize}

\end{example}

Our first main result is an algorithm for direct access for lexicographic orders, including ones that are not achievable by  past structures. 
We further show that within the class of CQs without self-joins,
our algorithm covers all the tractable cases (in the sense adopted here),
and we
establish a decidable and easy-to-test classification
of the lexicographic orders over the free variables into tractable
and intractable ones.  For instance, in
the case of $\rel{Visits}\Join \rel{Cases}$ the lexicographic order
$(\att{\#cases}, \att{age}, \att{city},\att{date},\att{person})$
is intractable.  It is classified as such because $\att{\#cases}$ and
$\att{age}$ are non-neighbours (i.e., do not co-occur in the same
atom), but $\att{city}$, which comes after $\att{\#cases}$ and
$\att{age}$ in the order, is a neighbour of both.  This is what we
call a \e{disruptive trio}.
The lexicographic order
$(\att{\#cases}, \att{age})$ is also intractable since the query
$\rel{Visits}\Join \rel{Cases}$ is not
$\set{\att{\#cases}, \att{age}}$-connex.  
In contrast, the
lexicographic order $(\att{\#cases},\att{city},\att{age})$ is
tractable.  We also show that within the tractable side, the structure
we construct
allows for inverted access
in constant time.

Our classification is proved in two steps.  We begin by considering
the complete lexicographic orders (that involve all free variables). 
We
show that for free-connex CQs without self-joins, 
the absence of a disruptive trio is a sufficient and
necessary condition for tractability. 
We then generalize to partial
lexicographic orders $L$
where the ordering is determined only by
a \emph{subset of the free variables}. 
There, the
condition is that there is no disruptive trio \e{and} that the query
is $L$-connex
(a similar condition to being free-connex,
but for the subset of the variables that appear in $L$ instead of the free ones). 
Interestingly, it turns out
that a partial lexicographic order is tractable if and only if
it is the prefix of a complete tractable lexicographic order.

Next, we study the selection problem for lexicographic orders and show that being free-connex  is a sufficient and necessary condition for a linear-time solution in the case of CQs without self-joins. 
In particular, there are ordered queries  with a tractable selection but intractable direct access, namely free-connex CQs without self-joins where we have a disruptive trio  
or lack the property of being $L$-connex.

A lexicographic order is a special case of an ordering by the \e{sum
of attribute weights}, where every database value is mapped to some weight. 
Hence, a natural question is which CQs have a tractable direct
access by the order of sum.  For example, what about
$\rel{Visits}\Join \rel{Cases}$ with the order
$(\alpha{\cdot}\att{\#cases}+\beta{\cdot}\att{age})$?  It is easy to
see that this order is intractable because the lexicographic order
$(\att{\#cases}, \att{age})$ is intractable.  In fact, it is easy to
show that an
order by sum is intractable whenever \e{there exists} an intractable
lexicographic order (e.g., there is a disruptive trio).
However, we will show that the situation is worse: the only tractable case is the
one where the CQ is acyclic and there is an atom that contains all of
the free variables. 
In particular, ranked direct access by sum is intractable
already for the Cartesian product
$Q(c_1,d,x,p,a,c_2)
  \datarule\rel{Visits}(p,a,c_1),\rel{Cases}(c_2,d,x)$, even though
\e{every} lexicographic order is tractable (according to our
aforementioned classification).
This daunting hardness also emphasizes how ranked direct access is fundamentally harder than \e{ranked enumeration} where, in the case of the sum of attributes, the answers of \emph{every free-connex}
CQ can be enumerated with logarithmic delay after a 
linear
preprocessing time~\cite{tziavelis22anyk}.

Next, we study the selection problem for the sum of weights, and establish the following dichotomy in complexity (again
assuming fine-grained hypotheses): the selection problem can be solved
in $\bigO(n\log n)$ time, where $n$ is the size of the database, if and
only if the hypergraph of the CQ restricted to the free variables contains at most two maximal
hyperedges (w.r.t.~containment).
The tractable side is applicable even in the presence of self-joins, and it is achieved by adopting an algorithm
by Frederickson and Johnson~\cite{frederickson84selection}
originally developed for selection on sorted matrices.
For illustration, the selection problem is solvable in quasilinear time
for the query $\rel{Visits}\Join \rel{Cases}$ ordered by sum.

Lastly, we study the implication of FDs on our results, and generalize all of them to incorporate a set of unary FDs (i.e., FDs with a single attribute on the left-hand side).
Like previous works on FDs on
enumeration~\cite{DBLP:journals/mst/CarmeliK20},
deletion propagation~\cite{DBLP:conf/pods/Kimelfeld12},
resilience~\cite{DBLP:journals/pvldb/FreireGIM15},
and probabilistic inference~\cite{DBLP:journals/pvldb/GatterbauerS15}, we use the notion of an extended CQ to reason about the tractability of a CQ under the presence of FDs.
The idea is that by looking at the structure of the extended CQ (without FDs),
then we are able to classify the original CQ together with the FDs.
While this works in a relatively straightforward way for the case of sum,
the case of lexicographic orders is more involved since the FDs may interact with the order in non-trivial ways.
To extend our dichotomy results for lexicographic orders to incorporate FDs,
we show how the extension of a CQ and order may also result in a \e{reordering} of the variables.
Then, tractability is decided by the extended CQ together with the reordered lexicographic order.

\paragraph{Overview of results} We summarize our results 
(excluding the dichotomies under the presence of FDs) in
\Cref{fig:overview} with different colors indicating the tractability
of the studied orders, namely
lexicographic (LEX) and sum-of-weights (SUM) orders. 
For both direct access and selection, we obtain the precise
picture of the orders and CQs without self-joins that are tractable according to our yardstick: $\bigO(n \polylog n)$ preprocessing and $\bigO(\polylog n)$ per access for direct access
(conveniently denoted as $\langle n \polylog n, \polylog n \rangle$
for $\langle \textrm{preprocessing}, \textrm{access} \rangle$)
and $\bigO(n \polylog n)$ for selection (or $\Comp{1}{n \log n}$).
Finally, we show how the results are
affected by every set of unary FDs (not depicted in \Cref{fig:overview}); 
in other words, we extend our
dichotomies to incorporate the FDs of the underlying schema under the
restriction that all FDs have a single attribute on the premise
(leaving the case of the more general FDs open for future research).

\paragraph{Comparison to an earlier conference version} A
preliminary version of this manuscript appeared
in a conference proceedings~\cite{confversion}. Compared to that
version, this manuscript includes several significant extensions and
improvements. First, we added an investigation on the complexity
of the selection problem with lexicographic orders, establishing a
complete dichotomy result (\Cref{thm:lex-selection-dicotomy} and all
of \Cref{sec:lex_selection}). Second, we extended a dichotomy
from the conference version to include self-join-free CQs \e{beyond
  full CQs} for the selection problem by SUM, so it now covers all
self-join-free CQs (with projections), thereby resolving the
corresponding open question from the conference paper
(\Cref{sec:sum_selection_projections}).  Third, we extended our
results to cover unary FDs (\Cref{sec:FDs}).  Fourth, we have
clarified the relationship between the disruptive trio and the concept
of an elimination order (\Cref{rem:elim_order}).  Fifth and last, we
made considerable simplifications and improvements in previous
components, including the proof of hardness of direct access for
lexicographic orders (\Cref{lemma:hardness})
and several proofs for the
SUM selection (\Cref{sec:sum_selection}).

\paragraph{Applicability}
It is important to note that while our results are stated over a limited class of queries (a fragment of acyclic CQs), there are some implications beyond this class that are immediate yet significant. In particular, we can use known techniques that reduce other CQs to a tractable form and then apply our direct-access or selection algorithms. 
As an example, a \emph{hypertree decomposition} can be used to transform a cyclic CQ to an acyclic form by paying a non-linear overhead during preprocessing~\cite{GottlobGLS:2016}.
As another example, a CQ with inequality ($<$) predicates can be reduced to a CQ
without inequalities
by paying only a polylogarithmic-factor increase in the size of the database~\cite{tziavelis21inequalities}.

\paragraph{Outline} The remainder of the manuscript is organized as
follows.  \Cref{sec:prelim} gives the necessary background.
In
\Cref{sec:lexic}, we consider direct access by lexicographic orders
that include all the free variables, and \Cref{sec:partial} extends
the results to partial ones.  We move on to the (for the most part)
negative results for direct access by sum orders in
\Cref{sec:sum_direct}. We study the selection problem for
lexicographic orders and sum in \Cref{sec:lex_selection} and
\Cref{sec:sum_selection}, respectively.  We extend our results to
incorporate unary FDs in \Cref{sec:FDs} and, lastly, conclude and
give future directions in \Cref{sec:conclusions}.

\section{Preliminaries}
\label{sec:prelim}

\subsection{Basic Notions}
\label{sec:basicNotions}

\introparagraph{Database}
A \e{schema} $\calS$ is a set of relational symbols $\{ R_1, \ldots, R_m \}$.
We use $\ar(R)$ for the arity of a relational symbol $R$.
A \e{database instance} $I$ 
contains a finite relation $R^I \subseteq \dom^{\ar(R)}$ for each $R \in \calS$, 
where $\dom$ is a set of constant values called the \e{domain}.
When $I$ is clear, we simply use $R$ instead of $R^I$.
We use $n$ for the size of the database, i.e., 
the total number of tuples.

\introparagraph{Queries}
A \e{conjunctive query} (CQ) $Q$ over schema $\calS$ is an expression of the form
$Q(\mathbf{X}_f) \datarule R_1(\mathbf{X}_1), \ldots, R_\ell(\mathbf{X}_\ell)$,
where the tuples $\mathbf{X}_f, \mathbf{X}_1, \ldots, \mathbf{X}_\ell$ hold variables,
every variable in $\mathbf{X}_f$ appears in some $\mathbf{X}_1, \ldots, \mathbf{X}_\ell$,
and $\{ R_1, \ldots, R_\ell \} \subseteq \calS$.
Each $R_i(\mathbf{X}_i)$ is called an \e{atom} of the query $Q$, 
and the set of all atoms is denoted by $\atoms(Q)$.
We use $\var(e)$ 
or $\var(Q)$ for the set of variables that appear in an atom $e$ or query $Q$, respectively.\footnote{We use $e$ for atoms because of the natural analogy to hyperedges in hypergraphs associated with a query $Q$.}
The variables $\mathbf{X}_f$ are called \e{free} and are denoted by $\free(Q)$.
A CQ is \e{full} if $\free(Q) = \var(Q)$ and \e{Boolean} 
if $\free(Q) = \emptyset$.
Sometimes, we say that CQs that are not full have \e{projections}.
A repeated occurrence of a relational symbol is a \e{self-join} and if no self-joins exist, a CQ is called \e{self-join-free}.
A homomorphism 
from a CQ $Q$ to a database $I$ 
is a mapping of $\var(Q)$ to constants from $\dom$,
such that every atom of $Q$ maps to a tuple in the database $I$.
A \e{query answer} is such a homomorphism 
followed by a projection 
on
the free variables.
The answer to a Boolean CQ is whether such a homomorphism exists.
The set of query answers is $Q(I)$ and we use $q \in Q(I)$ for a query answer.
For an atom $R(\mathbf{X})$ of a CQ, 
we say that a tuple $t \in R$ assigns a variable $x$ to value $c$ and denote it as $t[x] = c$
if for every index $i$ such that $\mathbf{X}[i] = x$ we have that $t[i] = c$.
The \e{active domain} of a variable $x$ is the subset of $\dom$ that $x$ can be assigned from the database $I$.

\introparagraph{Hypergraphs}
A {\em hypergraph} $\calH=(V,E)$ is a set $V$ of {\em vertices} and a set $E$ of subsets of $V$ called {\em hyperedges}.
Two vertices in a hypergraph are {\em neighbors} if they appear in the same edge.
A {\em path} of $\calH$ is a sequence of vertices such that every two succeeding vertices are neighbors.
A {\em chordless path} is a path in which no two non-succeeding vertices appear in the same hyperedge 
(in particular, no vertex appears twice).
A \emph{join tree} of a hypergraph $\calH=(V,E)$ is a tree $T$ where the nodes\footnote{
To make a clear distinction between the vertices of a hypergraph and those of its join tree, we call the latter nodes.}
are the hyperedges of $\calH$
and the {\em running intersection} property holds, namely: 
for all $u \in V$ the set $\{e \in E \mid u \in e\}$ forms a (connected) subtree in $T$.
An equivalent phrasing of the running intersection property is that 
given two nodes $e_1,e_2$ of the tree, 
for any node $e_3$ on the simple path between them, we have that $e_1\cap e_2\subseteq e_3$.
A hypergraph $\calH$ is {\em acyclic} if there exists a join tree for $\calH$.
We associate a hypergraph $\calH(Q) = (V, E)$ to a CQ $Q$ 
where the vertices are the variables of $Q$, 
and every atom of $Q$ corresponds to a hyperedge with the same set of variables.
Stated differently, $V = \var(Q)$ and
$E = \{ \var(e) | e \in \atoms(Q) \}$.
With a slight abuse of notation, we identify atoms of $Q$ with hyperedges of $\calH(Q)$.
A CQ $Q$ is {\em acyclic} if $\calH(Q)$ is acyclic,
otherwise it is \e{cyclic}.
The free-restricted hypergraph $\calH_{\free}(Q)$ is the restriction of $\calH(Q) = (V, E)$ on the
nodes that correspond to free variables, i.e., 
$\calH_{\free}(Q) = (\free(Q), \{e \cap \free(Q) | e \in E\})$.

\introparagraph{Free-connex CQs}
A hypergraph $\calH'$ is an \emph{inclusive extension} of $\calH$ if every edge of $\calH$ appears in $\calH'$, and every edge of $\calH'$ is a subset of some edge in $\calH$.
Given a subset $S$ of the vertices of $\calH$, a tree $T$ is an 
\emph{ext-$S$-connex tree} (i.e., extension-$S$-connex tree)
for a hypergraph $\calH$ if:
(1) $T$ is a join tree of an inclusive extension of $\calH$, and
(2) there is a subtree\footnote{By subtree, we mean any connected subgraph of the tree.} $T'$ of $T$ that contains exactly the vertices $S$~\cite{bdg:dichotomy}.
We say that a hypergraph is $S$-connex if it has an ext-$S$-connex tree~\cite{bdg:dichotomy}.
A hypergraph is $S$-connex iff it is acyclic and it remains acyclic after the addition of a hyperedge containing exactly $S$~\cite{bb:thesis}.
Given a hypergraph $\calH$ and a subset $S$ of its vertices, an {\em $S$-path} is a chordless path $(x,z_1,\ldots,z_k,y)$ in $\calH$ with $k\geq 1$, such that $x,y\in S$, and $z_1,\ldots,z_k\not\in S$.
A hypergraph is $S$-connex iff
it has no $S$-path~\cite{bdg:dichotomy}.
A CQ $Q$ is {\em free-connex} if $\calH(Q)$ is $\free(Q)$-connex~\cite{bdg:dichotomy}.
Note that a free-connex CQ is necessarily acyclic.\footnote{Free-connex CQs are sometimes called in the literature \emph{free-connex acyclic} CQs~\cite{bdg:dichotomy}. As free-connexity is not defined for cyclic CQs, we choose to omit the word \emph{acyclic} and simply call these CQs \emph{free-connex}.}
An implication of the characterization given above is that it suffices to find a join-tree for an inclusive extension of a hypergraph $\calH$ to infer that $\calH$ is acyclic.

To simplify notation, we also say that a CQ is $L$-connex for a (partial) lexicographic order $L$ if the CQ is $S$-connex for the set of variables $S$ that appear in $L$.
Generalizing the notion of an inclusive extension,
we say that a hypergraph $\calH'$ is \emph{inclusion equivalent} to $\calH$ if every hyperedge of $\calH$ is a subset of some hyperedge of $\calH'$ and vice versa.
For example, the two hypergraphs with hyperedges 
$E_1 = \{\{x, y\}, \{y, z\}\}$ and $E_2 = \{\{x, y\}, \{y, z\}, \{z\}\}$
are inclusion equivalent because
$\{z\}$ is a subset of $\{y, z\}$ and every hyperedge is trivially a subset of itself.

\subsection{Problem Definitions}

\introparagraph{Orders over Answers}
For a CQ $Q$ and database instance $I$, 
we assume a \e{total order} $\preceq$ over the query answers $Q(I)$.
We consider two types of orders in this paper:

\begin{enumerate}
	\item \textbf{LEX}:
	Assuming that the domain values are ordered, a \e{lexicographic order} $\lex$ is an ordering of $\free(Q)$ 
	such that $\preceq$ compares two query answers $q_1, q_2$ on the value of the first variable in $\lex$,
	then on the second (if they are equal on the first), 
	and so on~\cite{harzheim06ordered}. 
	A lexicographic order is called \e{partial} if the variables in $\lex$ are a subset of $\free(Q)$.

	\item \textbf{SUM}:
	The second type of order assumes
	given weight functions that assign real-valued weights to the domain values of each variable.
	More precisely, for each variable $x$, we define a function $w_x: \dom \rightarrow \R$.
	Then, the query answers are ordered by a weight which is computed 
	by aggregating the weights of the assigned values of free variables.
	In a \e{sum-of-weights order}, denoted by SUM, 
	we have the weight of each query answer $q \in Q(I)$ to be $w_Q(q) = \sum_{x \in \free(Q)} w_x(q(x))$ and 
    $q_1 \preceq q_2$ implies that $w_Q(q_1) \leq w_Q(q_2)$.
    We emphasize that we allow only free variables to have weights,
    otherwise different semantics for the query answers are possible \cite{Tziavelis:fullversion}.
	To simplify notation, we sometimes refer to all $w_x, x \in \free(Q)$ 
	and $w_Q$ together as one weight function $w$.
\end{enumerate}

\introparagraph{Attribute Weights vs.\ Tuple Weights for SUM}
Notice that in the definition above, 
we assume that the input weights are assigned to the domain values of the attributes.
Alternatively, the input weights could be assigned to the relation tuples, 
a convention that has been used in past work on ranked enumeration \cite{tziavelis20vldb}.
Since there are several reasonable semantics for interpreting a tuple-weight ranking for CQs with projections and/or self-joins \cite{Tziavelis:fullversion},
we elect to present our results for the case of attribute weights.
We note that our results directly extend to the case of tuple weights
for full self-join-free CQs where the semantics are clear.
On the one hand,
attribute weights can easily be transformed to tuple weights in linear time such that the weights of the query answers remain the same.
This works by assigning each variable to one of the atoms that it appears in,
and computing the weight of a tuple by aggregating the weights of the assigned attribute values.
Therefore, our hardness results for SUM orders directly extend to the case of tuple weights.
On the other hand, our positive results on
direct access (\Cref{sec:sum_direct}),
selection (\Cref{sec:sum_selection_tractable})
and their extension to the case of FDs (\Cref{sec:fds_sum})
rely on algorithms that innately operate on tuple weights,
thus we cover those cases too.

\introparagraph{Direct Access vs.\ Selection}
We now define two problems that both directly access ordered query answers.
Since our goal is to classify the combination of CQs and orders by their tractability, 
we let those two define the problem.
Specifically, a problem is defined by a CQ $Q$ and a family of orders $\Pi$.
The reason that we use a family of orders in the problem definition is that 
for the case of SUM,
we do not distinguish between different weight functions in our classification.
For LEX, we always consider the family of orders to contain only one specific (partial) lexicographic order.

\begin{definition}[Direct Access]
    Let $Q$ be a CQ and $\Pi$ a family of total orders.
    The problem of \e{direct access} by $\Pi$ takes as an input a database $I$ and an order $\prec$ from $\Pi$
    and constructs a data structure (called the \e{preprocessing} phase)
    which then allows access to a query answer $q \in Q(I)$ at any index $k$ of the
    (implicit) array of query answers sorted by $\prec$.
\end{definition}

The essence of direct access is that after the preprocessing phase, we need to be able to support
multiple such accesses.
Notably, the values of $k$ that are going to be requested afterward are not known during preprocessing.

\begin{definition}[Selection]
    Let $Q$ be a CQ and $\Pi$ a family of total orders.
    The problem of \e{selection} by $\Pi$ takes as an input a database $I$,
    an order $\prec$ from $\Pi$,
    and asks for the query answer $q \in Q(I)$ at index $k$ of the
    (implicit) array of query answers sorted by $\prec$.
\end{definition}	

The problem of \e{selection} \cite{blum73select,floyd75select,frederickson93select}
is a computationally easier task that requires only a single direct access, 
hence does not make a distinction between preprocessing and access phases.
A special case of the problem is finding the median query answer.

For both problems, if the index $k$ exceeds the total number of answers, 
the returned answer is ``out-of-bound''.

\subsection{Complexity Framework and Sorting}
\label{sec:complexity}

We measure asymptotic complexity in terms of the size of the database $n$,
while the size of the query is considered a constant.
If the time for preprocessing is $\bigO\big(f(n)\big)$ and the time for each access is $\bigO\big(g(n)\big)$,
we denote that as $\Comp{f(n)}{g(n)}$,
where $f,g $ are functions from $\N$ to $\R$.
Note that by definition, the problem of selection asks for a $\Comp{1}{g(n)}$ solution.

Our goal for both problems is to achieve efficient access in time significantly smaller than (the worst case) $|Q(I)|$.
For direct access, we consider the problem tractable if $\Comp{n \log n}{\log n}$ is possible,
and for selection $\Comp{1}{n \log n}$.

The model of computation is the standard RAM model with uniform cost measure.
In particular, it allows for linear time construction of lookup tables,
which can be accessed in 
constant time.
We would like to point out that some past works \cite{bdg:dichotomy,oldRandomAccess} have assumed 
that in certain variants of the model, sorting 
can be done in linear time \cite{grandjean96sorting}.
Since we consider problems related to summation and sorting \cite{frederickson84selection} 
where a linear-time sort would improve otherwise optimal bounds,
we adopt a more standard assumption that sorting is comparison-based and possible only in quasilinear time.
As a consequence, some upper bounds mentioned in this paper
are weaker than the original sources which assumed linear-time sorting \cite{bb:thesis,oldRandomAccess}.

\subsection{Hardness Hypotheses}
\label{sec:hypotheses}

All the lower bounds we prove are conditional on one or multiple of the following four hypotheses.

\begin{hypothesis}[\sparseBMM]
	Two Boolean matrices $A$ and $B$, represented as lists of their non-zero entries, cannot be multiplied in time $m^{1+o(1)}$, 
	where $m$ is the number of non-zero entries in $A$, $B$, and $AB$.
\end{hypothesis}

A consequence of this hypothesis is that we cannot answer the query 
$Q(x,z) \datarule R(x,y), S(y,z)$
with quasilinear preprocessing and polylogarithmic delay.
In more general terms, any self-join-free acyclic non-free-connex CQ 
cannot be enumerated with quasilinear\footnote{
\label{note:linear}
Works in the literature~\cite{bagan2008computing,berkholz2020tutorial,oldRandomAccess}
typically phrase this as linear,
yet any logarithmic factor increase is still covered by the hypotheses.
}
preprocessing time and polylogarithmic delay assuming the \sparseBMM{} hypothesis~\cite{bdg:dichotomy, berkholz2020tutorial}.

\begin{hypothesis}[\hyperclique~\cite{abboud14conjectures,DBLP:conf/soda/LincolnWW18}]
For every $k \geq 2$, there is no
$O(m \polylog m)$ algorithm for deciding the existence of a
$(k{+}1,k)$-hyperclique in a $k$-uniform hypergraph with $m$ hyperedges,
where a \e{$(k{+}1,k)$-hyperclique} is a set of $k{+}1$ vertices
such that every subset of $k$ elements is a hyperedge.
\end{hypothesis}

When $k=2$, 
this follows from the ``$\delta$-Triangle'' hypothesis \cite{abboud14conjectures}. 
This is the hypothesis that we cannot detect a triangle in a graph in linear time~\cite{Alon.1997}.
When $k\ge 3$, this is a special case of the ``$(\ell,k)$-Hyperclique''
hypothesis~\cite{DBLP:conf/soda/LincolnWW18}.
A known consequence is that Boolean cyclic and self-join-free CQs cannot be answered in 
quasilinear\cref{note:linear}
time~\cite{bb:thesis}.
As a result, cyclic and self-join-free CQs do not admit 
enumeration with quasilinear preprocessing time and polylogarithmic delay assuming the \hyperclique{} hypothesis~\cite{bb:thesis}.

\begin{hypothesis}[\ThreeSUM~\cite{patrascu2010dynamic,ilya053sum}]
Deciding whether there exist $a \in A, b \in B, c \in C$
from three sets of integers $A, B, C$, each of size $\Omega(m)$, such that 
$a + b + c = 0$ cannot be done in time $O(m^{2-\epsilon})$ for any $\epsilon > 0$.
\end{hypothesis}

In its simplest form, the \ThreeSUM{} problem 
asks for three distinct real
numbers $a, b, c$ from a set $S$ with $m$ elements 
that satisfy $a + b + c = 0$.
There is a simple $O(m^2)$ algorithm for the problem, but it is conjectured that in general, 
no truly subquadratic solution exists \cite{patrascu2010dynamic}.
The significance of this conjecture has been highlighted by many conditional lower bounds for problems 
in computational geometry \cite{gajentaan95geom} and 
within the class P in general \cite{williams153sum}.
Note that the problem remains hard even for integers provided that they are sufficiently large 
(i.e., in the order of $O(n^3)$) \cite{patrascu2010dynamic}.
The hypothesis we use here has three different sets of numbers, but it is equivalent~\cite{ilya053sum}.
This lower bound has been confirmed in some restricted models of computation \cite{erickson95sum,ailon05sum}.

\begin{hypothesis}[\seth~\cite{impagliazzo01seth}]
For the satisfiability problem with $m$ variables and 
$k$ variables per clause ($k$-SAT),
if $s_k$ is the infimum of the real numbers $\delta$
for which $k$-SAT admits an
$\bigO(2^{\delta m})$ algorithm,
then $\lim_{k \to \infty} s_k = 1$
\end{hypothesis}

Intuitively, the \emph{Strong Exponential Time Hypothesis} (\seth{}) states that the best possible algorithms for $k$-SAT
approach $\bigO(2^m)$ running time when $k$ goes to infinity.
\seth{} implies that the $k$-Dominating Set problem on a graph with $m$ vertices
cannot be solved in 
$\bigO(m^{2-\epsilon})$ for $k \geq 3$ and any constant $\epsilon$ \cite{patrascu10sat}.
Based on that, it can be shown that counting the answers to a self-join-free acyclic CQ that is not free-connex  
cannot be done in $\bigO(n^{2-\epsilon'})$ for any constant $\epsilon'$ \cite{stefan}.

\subsection{Known Results for CQs}\label{sec:past}

\introparagraph{Eliminating Projection}
We now provide some background that relates to the efficient handling of CQs.
For a query with projections, a standard strategy is to reduce it
to an equivalent one 
where techniques for acyclic full CQs can be leveraged.
The following proposition, which is widely known and used \cite{berkholz2020tutorial}, shows that
this is possible for free-connex CQs.

\begin{proposition}[Folklore]
\label{prop:reduce-to-subtree}
Given a database instance $I$, a CQ $Q$, 
a join tree $T$ of an inclusive extension of $Q$,
and a subtree $T'$ of $T$ that contains all the free variables,
we can compute 
in linear time a database instance $I'$ 
over the schema of a CQ $Q'$ that consists of the nodes of $T'$ such that 
$Q(I)=Q'(I')$
and $|I'| \leq |I|$.
\end{proposition}

\noindent
This reduction is done by first creating a relation for every node in $T$ using projections of existing relations, then performing the classic semi-join reduction by Yannakakis~\cite{Yannakakis} to filter the relations of $T'$ according to the relations of $T$, and then we can simply ignore all relations that do not appear in $T'$ and obtain the same answers.
Afterward, they can be handled efficiently, 
e.g. their answers can be enumerated with constant delay \cite{bdg:dichotomy}.
We refer the reader to recent tutorials~\cite{berkholz2020tutorial,durand20tutorial} for an intuitive illustration of the idea.

\introparagraph{Ranked enumeration}
Enumerating the answers to a CQ in ranked order is a special case of direct access 
where the accessed indexes are consecutive integers starting from $0$.
As it was recently shown~\cite{tziavelis20tutorial}, ranked enumeration for CQs is intimately connected to classic algorithms on finding the $k^\textrm{th}$ shortest path in a graph.
In contrast to direct access, ranked enumeration by SUM orders
(which also includes lexicographic orderings as a special case) 
is possible with logarithmic delay after a linear-time preprocessing phase
for all free-connex CQs \cite{tziavelis20vldb}.
In contrast, as we will show, that is not the case for direct access. 
Existing ranked-enumeration algorithms rely on priority queue structures 
that compare a minimal number of candidate answers
to produce the ranked answers one-by-one in order.
There is no straightforward way to extend them to the task of direct access where we may 
skip over a large number of answers to get to an arbitrary index $k$.

\introparagraph{Direct Access}
Carmeli et al.~\cite{oldRandomAccess} devise a direct access structure (called ``random access'') and use it to uniformly sample CQ answers (called ``random-order enumeration''). 
While it leverages the idea of using count statistics on the input tuples to navigate the space of query answers that had also been used in prior work on sampling~\cite{Zhao18sampling},
it decouples it from the random order requirement and advances it into direct access.
The separation into a direct access component and a random permutation (of indices) generated externally also allows sampling without replacement which was not possible before.
This direct access algorithm is also a significant simplification over a prior one by Brault-Baron~\cite{bb:thesis}. 
We emphasize that even though these algorithms do not explicitly discuss the order of the answers, 
a closer look shows that they internally use and produce \e{some lexicographic order}.

\begin{theorem}[\cite{oldRandomAccess,bb:thesis}]\label{theorem:known-access}
Let $Q$ be a CQ. If $Q$ is free-connex, then direct access (in some order) 
is possible 
in $\Comp{n \log n}{\log n}$.
Otherwise, if it is also self-join-free, then direct access (in any order) 
is not possible 
in $\Comp{n \polylog n}{\polylog n}$,
assuming $\sparseBMM$ and $\hyperclique$.
\end{theorem}

Recent work by Keppeler~\cite{Keppeler2020Answering} 
suggests another 
direct-access solution
by lexicographic order, 
which also supports efficient 
insertion and deletion of input tuples.
Given these additional requirements,
the supported CQs are more limited, 
and are only
a subset of free-connex CQs called \emph{$q$-hierarchical}~\cite{Berkholz17updates}. This is a subclass of the well-known \emph{hierarchical} queries with an additional restriction on the
existential
variables.
As an example, the following CQs are not $q$-hierarchical even though they are free-connex:
$Q_1(x,y) \datarule R_1(x),R_2(x,y),R_3(y)$ and
$Q_2(x) \datarule R_1(x,y),R_2(y)$. 
For these queries, direct access is not supported by the solution of Keppeler~\cite{Keppeler2020Answering}, 
even though it is possible without the 
update requirements (as we show in \Cref{sec:lexic}).

All previous direct-access solutions of which we are aware have two gaps compared to this work: (1) they do not discuss which lexicographic orders (given by orderings of the free variables) are supported; (2) they do not support all possible lexicographic orders.
We conclude this section with a short survey of  existing solutions and their supported orders.

All 
prior
direct-access solutions use some component that depends on the query structure and constrains the supported orders.
The algorithm of Carmeli et al.~\cite[Algorithm 3]{oldRandomAccess} assumes that 
a join tree is given with the CQ, 
and the 
lexicographic
order is \emph{imposed by the join tree}. 
Specifically, it is 
an ordering of the variables
achieved by a preorder
depth-first traversal of the tree. 
As a result, it does not support 
any
order that requires 
jumping
back-and-forth
between
different
branches of the tree. 
In particular, it does not support  $Q_3(v_1,v_2,v_3,v_4)\datarule R(v_1,v_3), S(v_2,v_4)$ 
with the lexicographic order given by the increasing variable indices 
(we adopt this convention for all the examples below).
We show how to handle this CQ and order in detail in \Cref{ex:layers}.
The algorithm of Brault-Baron~\cite[Algorithm 4.3]{bb:thesis} assumes that an \emph{elimination order} is given along with the CQ. 
The resulting lexicographic order is affected by that elimination order, 
but is not exactly the same. 
This solution
suffers from similar restrictions, and
it does not support $Q_3$ either.
The algorithm of Keppeler~\cite{Keppeler2020Answering} assumes that a \emph{$q$-tree} is given with the CQ, and the possible 
lexicographic orders
are affected by this tree. %
Unlike the two earlier mentioned approaches,
this algorithm can interleave variables from different atoms, 
yet cannot support some orders that are possible for the previous algorithms. 
As an example, it does not support $Q_4(v_1,v_2,v_3) \datarule R_1(v_1,v_2),R_2(v_2,v_3)$ as $v_2$ 
is
highest in the hierarchy (the atoms containing it strictly subsume the atoms containing any other variable) and so it 
is necessarily the first variable in the q-tree and in the ordering produced.

Finally, we should mention that there exist queries and orders 
that
require \emph{both} jumping back-and-forth in the join tree \emph{and} visiting the variables in an order different than any hierarchy.
As a result, these are not supported by any previous solution. Two such examples are 
$Q_5(v_1,v_2,v_3,v_4,v_5) \datarule R_1(v_1,v_3),R_2(v_3,v_4),R_3(v_2,v_5)$ and 
$Q_6(v_1,v_2,v_3,v_4,v_5) \datarule R_1(v_1,v_2,v_4),R_2(v_2,v_3,v_5)$.
In \Cref{sec:lexic}, we provide an algorithm that supports both of these CQs.

\section{Direct Access by Lexicographic Orders}
\label{sec:lexic}

In this section, we answer the following question: 
for which underlying
lexicographic orders 
can we achieve 
``tractable'' direct access to ranked CQ answers, i.e.\  
with quasilinear preprocessing and polylogarithmic time per answer?

\begin{example}[No direct access]\label{ex:lex-alg}
Consider the lexicographic order 
$L = \angs{v_1, v_2, v_3}$ for the query
$Q(v_1, v_2, v_3) \datarule R(v_1, v_3), S(v_3, v_2)$.
Direct access to the query answers according to that order 
would allow us to 
 ``jump over'' the $v_3$ values via binary search 
and essentially enumerate the answers to
$Q'(v_1, v_2) \datarule R(v_1, v_3), S(v_3, v_2)$. 
However, we know that $Q'$ is not free-connex and that is impossible to achieve 
enumeration with quasilinear preprocessing and polylogarithmic delay (if \sparseBMM{} holds).
Therefore, the bounds we are hoping for are out of reach for the given query and order.
The core difficulty is that the joining variable $v_3$ appears \e{after} the other two in the lexicographic order.
\end{example}

We formalize this notion of ``variable in the middle'' in order to detect similar situations in more complex queries.

\begin{definition}[Disruptive Trio]
\label{def:trio}
Let $Q$ be a CQ and $\lex$ a lexicographic order of its free variables.
We say that three free variables $v_1, v_2, v_3$ 
are a \emph{disruptive trio} in $Q$ with respect to $\lex$ 
if $v_1$ and $v_2$ are not neighbors (i.e.\ they do not appear together in an atom), 
$v_3$ is a neighbor of both $v_1$ and $v_2$, 
and $v_3$ appears after $v_1$ and $v_2$ in $\lex$.
\end{definition}

As it turns out, 
for
free-connex and self-join-free CQs, the tractable CQs are precisely captured by this simple criterion.
Regarding self-join-free CQs that are not free-connex, their known  intractability of enumeration
implies that direct access is also intractable.
This leads to the following dichotomy:

\begin{theorem}[Direct Access by LEX]\label{thm:lex-dicotomy}
Let $Q$ be a CQ and $\lex$ be a lexicographic order.
\begin{itemize}
    \item If $Q$ is free-connex and does not have a disruptive trio with respect to $\lex$, then direct access by $\lex$ is possible 
    in $\Comp{n \log n}{\log n}$.
    \item Otherwise, if $Q$ is also self-join-free, then direct access by $\lex$ is not possible 
    in $\Comp{n \polylog n}{\polylog n}$
    assuming \sparseBMM{} and \hyperclique{}.
\end{itemize}
\end{theorem}

\begin{remark}
\label{rem:elim_order}
Assume we are given a full CQ, and the lexicographic order we want to achieve is $\angs{v_1, \ldots, v_m}$.
It was shown (in the context of ranked enumeration by lexicographic orders)
that the absence of
disruptive trios is equivalent 
to the 
existence of a reverse ($\alpha$-)elimination order
of the variables~\cite[Theorem 15]{bb:thesis}.
That is, we need there to exist an atom that contains $v_m$ and all of its neighbors (variables that share an atom with $v_m$), 
and if we remove $v_m$ from 
the query, $v_1,\ldots,v_{m-1}$ should recursively be a reverse elimination order. 
For the base case, when $m=1$, $v_1$ constitutes a reverse-elimination order.
\end{remark}

\begin{remark}
On the positive side of \Cref{thm:lex-dicotomy}, the preprocessing time is dominated by sorting the input relations, which we assume requires $\bigO(n \log n)$ time.
If we assume instead that sorting takes linear time 
(as assumed in some related work~\cite{bb:thesis,oldRandomAccess,grandjean96sorting}), 
then the time required for preprocessing is only $\bigO(n)$ instead of $\bigO(n \log n)$.
\end{remark}

In \Cref{sec:layered-alg}, we provide an algorithm for this problem for full acyclic CQs that have a particular join tree that we call \e{layered}. 
Then, we show how to find such a layered join tree whenever there is no disruptive trio in \Cref{sec:trio-to-layered}. 
In \Cref{sec:CQ-to-full}, we explain how to adapt our solution for CQs with projections, 
and in \Cref{sec:dichotomy} we prove a lower bound which establishes that our algorithm applies to \e{all} cases where direct access is tractable.

\subsection{Layer-Based Algorithm}\label{sec:layered-alg}

Before we explain the algorithm, we first define 
one of its main components.

A \e{layered join tree} is a join tree
where each node 
belongs to a layer. 
The layer number is the last position of any of its variables in the lexicographic order.
Intuitively, ``peeling'' off the outermost (largest) layers must result in 
a valid join tree 
(for a hypergraph with fewer variables).
To find such a join tree for a CQ $Q$,
we may have to introduce hyperedges that are contained in those of $\calH(Q)$ 
(this corresponds to taking the projection of a relation)
or remove hyperedges of $\calH(Q)$ that are contained in others
(this corresponds to filtering relations that contain a superset of the variables).
Thus, we define the layered join tree with respect to a hypergraph that is \emph{inclusion equivalent}
(recall the definition of an inclusion equivalent hypergraph from \cref{sec:basicNotions}).

\begin{definition}[Layered Join Tree]
\label{def:layered_join_tree}
Let $Q$ be a full acyclic CQ, and let $\lex=\angs{v_1, \ldots, v_f}$ be a lexicographic order.
A \emph{layered join tree} for $Q$ with respect to $\lex$ is a join tree of 
a hypergraph that is inclusion equivalent to $\calH(Q)$
where 
($1$) every node $\vertex$ of the tree is assigned to layer $\max\{i\mid v_i\in \vertex\}$, 
($2$) there is exactly one node 
for each layer, 
and ($3$) for all $j \leq f$ the induced subgraph with only the nodes that belong to the first $j$ layers is a tree.
\end{definition}

\begin{figure}[t]
\centering
\begin{subfigure}{.47\linewidth}
    \centering
    \includegraphics[scale=0.5]{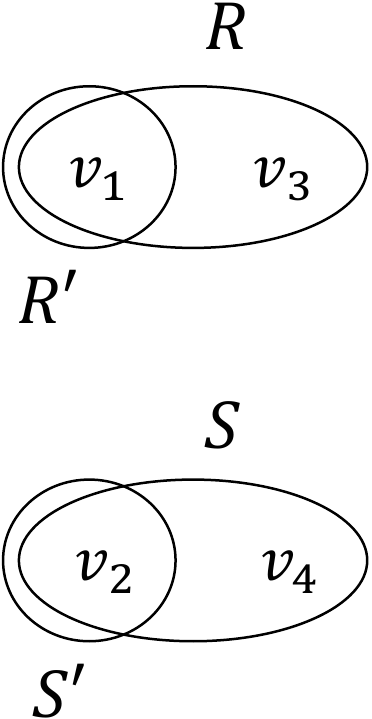}
    \caption{A hypergraph that is 
    inclusion equivalent to $\calH(Q_3)$.}
    \label{fig:InclExt}
\end{subfigure}%
\hfill
\begin{subfigure}{.47\linewidth}
    \centering
    \includegraphics[scale=0.5]{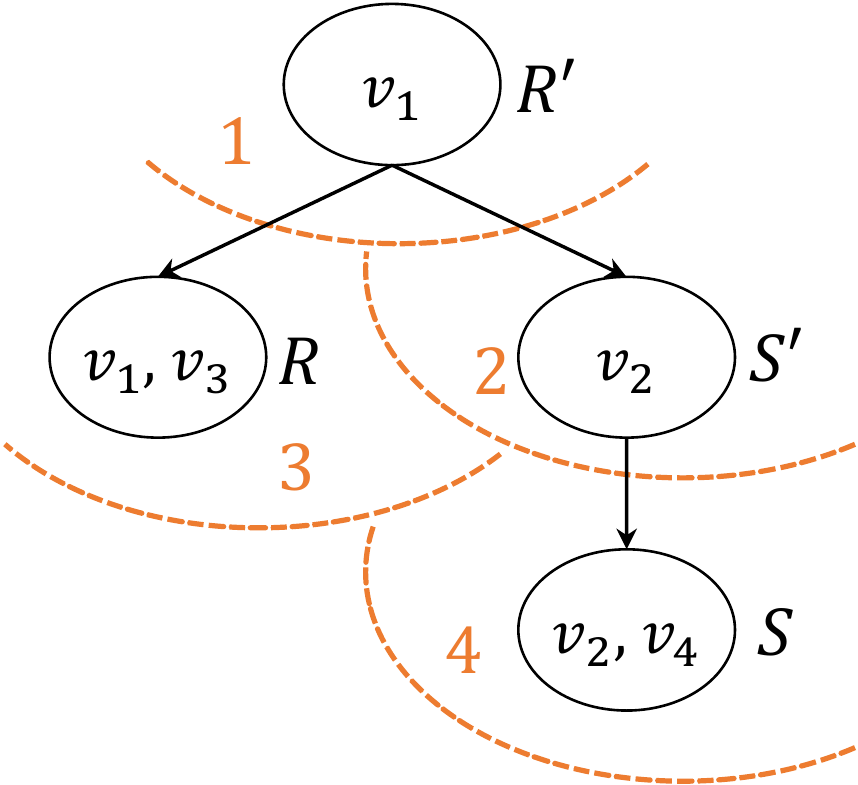}
    \caption{A layered join tree for $Q_3$ w.r.t. the lexicographic order.}
    \label{fig:layers}
\end{subfigure}
\caption{Constructing a layered join tree for the query $Q_3(v_1,v_2,v_3,v_4) \datarule R(v_1,v_3), S(v_2,v_4)$ and order $\angs{v_1,v_2,v_3,v_4}$.}
\label{fig:InclExt_layers}
\end{figure}

\begin{example}
\label{ex:layers}
Consider the CQ
$Q_3(v_1,v_2,v_3,v_4)\datarule R(v_1,v_3), S(v_2,v_4)$
and the lexicographic order $\angs{v_1,v_2,v_3,v_4}$.
To support that order, 
we first 
find an inclusion equivalent hypergraph,
shown in \Cref{fig:InclExt}.
Notice that we added two hyperedges that are strictly contained in the existing ones, and obtained a hypergraph corresponding to $R(v_1,v_3), R'(v_1),S(v_2,v_4),S'(v_2)$.
A layered join tree constructed from that hypergraph is depicted in \Cref{fig:layers}.
There are four layers, one for each node of the join tree.
The layer of the node containing $\{ v_1, v_3 \}$ is $3$ because 
$v_3$ appears after $v_1$ in the order and it is the third variable.
If we remove the last layer, then we obtain a layered join tree for the induced hypergraph where the last variable $v_4$ is removed. 
\end{example}

We now describe an algorithm that takes as an input a CQ $Q$, a lexicographic order $\lex$, and a corresponding layered join tree
and provides direct access to the query answers after a preprocessing phase.
For preprocessing, 
we leverage a construction from Carmeli et al.~\cite[Algorithm 2]{oldRandomAccess}
and apply it to our layered join tree.
For completeness, we briefly explain how it works below.
Subsequently, we describe the access phase that 
takes into account the layers of the tree 
to accommodate the provided lexicographic order.
We emphasize that the way we access the structure is different than that of the past work \cite{oldRandomAccess}, and that this allows support of lexicographic orders that were impossible for the previous access routine (e.g. the order in \Cref{ex:layers}).

\introparagraph{Preprocessing}
The preprocessing phase 
($1$) creates a relation for every node of the tree, 
($2$) removes dangling tuples, 
($3$) sorts the relations,
($4$) partitions the relations into buckets, 
and ($5$) uses dynamic programming on the tree to compute and store certain counts\footnote{The same count statistics are also leveraged in \cite[Sect.~4.2]{Zhao18sampling} in the context of sampling}.
After preprocessing, we are guaranteed that for all $i$, the node of layer $i$ has a corresponding relation where each 
tuple participates in at least one query answer; 
this relation is partitioned into buckets by the assignment of the variables preceding $i$. 
In each bucket, we sort the tuples lexicographically by $v_i$. 
Each tuple is given a weight that indicates the number of different answers this tuple agrees with when only joining 
its subtree.
The weight of each bucket is the sum of its tuple weights. 
We denote both by the function $\w$. 
Moreover, for every tuple $t$, we compute the sum of weights of the preceding tuples in the bucket, denoted by $\rng(t)$.
We use $\endIndex(t)$ for the 
sum that corresponds to
the tuple following $t$ in the same bucket; 
if $t$ is last, we set this to be the bucket weight.
If we think of the query answers in the subtree sorted in the order of $v_i$ values,
then $\rng$ and $\endIndex$ 
distribute the indices between $0$ and the bucket weight to tuples. 
The number of indices within the range of each tuple corresponds to its weight.

\begin{example}[Continued]
\label{ex:lex_preprocessing}
The result of the preprocessing phase on an example database for our query $Q_3$
is shown in \Cref{fig:lex_preprocessing}.
Notice that $R$ has been split into two buckets according to the values of its parent $R'$, 
one for value $a_1$ and one for $a_2$.
For tuple $(a_1) \in R'$, we have $\w((a_1)) = 8$ because this is the number of answers that agree on that value in its subtree:
the left subtree has $2$ such answers which can be combined with any of the $4$ possible answers of the right subtree.
The start index of tuple $(b_1, d_3) \in S$ is the sum of the previous weights within the bucket:
$\rng((b_1, d_3)) = \w((b_1, d_1)) + \w((b_1, d_2)) = 1 + 1 = 2$.
Not shown in the figure is that every bucket stores the sum of weights it contains.

\end{example}

\begin{figure}[t]
\centering
\includegraphics[scale=0.5]{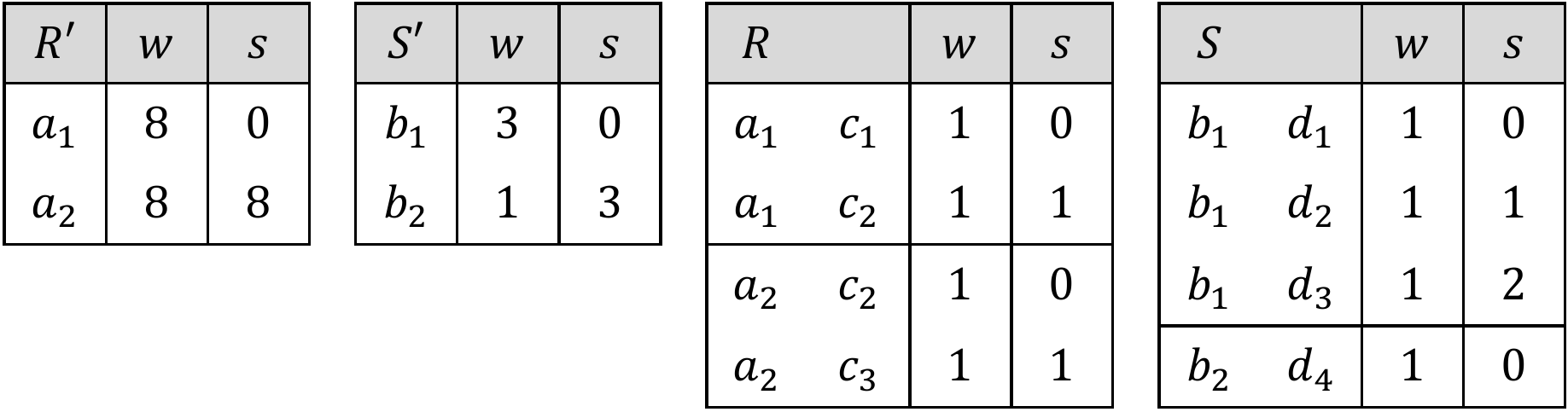}
\caption{\Cref{ex:lex_preprocessing}: The result of the preprocessing phase on $Q_3$, the layered join tree (\Cref{fig:layers}) and an example database. 
The weight and start index for each tuple are abbreviated in the figure as $w$ and $s$ respectively.}
\label{fig:lex_preprocessing}
\end{figure}

\introparagraph{Access}
The access phase works by going through the tree layer by layer.
When resolving a layer $i$, we select a tuple from its corresponding relation, 
which sets a value for the $i^{\textrm{th}}$ variable in $\lex$, 
and also determines a 
bucket for each child. 
Then, we conceptually erase the node of layer $i$ and its outgoing edges.

The access algorithm 
maintains a directed forest and an assignment to a prefix of the variables.
Each tree in the forest represents the answers obtained by joining its relations. 
Each root contains a single bucket that agrees with the already assigned values, 
thus every answer agrees on the prefix. 
Due to the running intersection property, 
different trees cannot share unassigned variables. 
As a consequence, any combination of answers from different trees can be 
added to the prefix assignment to form an answer to $Q$. 
The answers obtained this way are exactly the answers to $Q$ that agree with the already set assignment.
Since we start with a layered join tree, we are guaranteed that at each step, 
the next layer (which corresponds to the variable following the prefix for which we have an assignment) 
appears as a root in the forest.

Recall that from the preprocessing phase,
the weight of each root is the number of answers in its tree.
When we are at layer $i$, we have to take into account the weights of all the other roots in order to compute the number of query answers for a particular tuple.
More specifically,
the number of answers to $Q$ containing the already selected attributes (smaller than $i$) 
and some $v_i$ value contained in a tuple is 
found by multiplying the tuple weight with the weights of all other roots.
That is because 
the answers from all trees can be combined into a query answer.
Let $t$ be the selected tuple when resolving the $i^{\textrm{th}}$ layer. 
The number of answers to $Q$ that have a value of $\lex[i]$ smaller than that of $t$ and a value of $\lex[j]$ equal to that of $t$ for all $j<i$ is then:
\[\sum_{t'}{\left(\w(t')\prod_{r\in \text{roots}}{\w(r)}\right)}\]
where $t'$ ranges over tuples preceding $t$ in its bucket.
Denote by $\factor$ the product of all root weights.
Then we can rewrite as:
\begin{align*}
    \left(\sum_{t'}{\w(t')}\right)\left(\prod_{r\in \text{roots}}{\w(r)}\right)
    =
    \rng(t)\cdot\factor\,.
\end{align*}
Therefore, when resolving layer $i$ we select the last tuple $t$ such that the index we want to access is at least $\rng(t)\cdot\factor$.

\begin{algorithm}
  \If{$k\ge\w(\root)$}{
    \Return ``out-of-bound''
    }
  $\bucket[1] = \root$\;
  $\factor = \w(\root)$\;
  \For{i=1,\ldots,f}{
    $\factor =\factor / \w(\bucket[i])$\;
    pick $\tuple\in \bucket[i]$ s.t.
   ~$\rng(\tuple)\cdot\factor \leq k < \endIndex(\tuple)\cdot\factor$\;\label{line:select-tuple}
   $k = k - \rng(\tuple)\cdot\factor$\;
   \For{child $V$ of layer $i$}{
        get the bucket $b\in V$ agreeing with the selected tuples\;
        $\bucket[\layer(V)]=b$\;
        $\factor =\factor\cdot\w(b)$\;
   }
  }
  \Return the answer agreeing with the selected tuples
\caption{Lexicographic Direct-Access}
\label{alg:lex-random-access}
\end{algorithm}

\Cref{alg:lex-random-access} summarizes the process we described where $k$ is the index to be accessed and $f$ is the number of variables. 
Iteration $i$ resolves layer $i$.
Pointers to the selected buckets from the roots are kept in a $\bucket$ array. 
The product of the weights of all roots is kept in a $\factor$ variable.
In each iteration, the variable $k$ is updated to the index that should be accessed among 
the answers that agree with the already selected attribute values.
Note that $\bucket[i]$ is always initialized when accessed since layer $i$ is guaranteed to be a child of a smaller layer.

\begin{example}[Continued]
\label{ex:lex_access}
We demonstrate how the access algorithm works for index $k=12$.
When resolving $R'$, the tuple $(a_2)$ is chosen since $8\cdot 1 \le 12 < 16\cdot 1$; then, the single bucket in $S'$ and the bucket containing $a_2$ in $R$ are selected.
The next iteration resolves $S'$. When it reaches 
line~\ref{line:select-tuple}, 
$k = 12-8 = 4$ 
and $\factor = 2$.
As $0\cdot 2 \le 4 < 3\cdot 2$, the tuple $(b_1)$ is selected.
Next, $R$ is resolved, which we depict in \Cref{fig:lex_access}.
The current index is $k= 4-0 = 4$. 
The weights of the other roots (only $S$ here) gives us $\factor=3$.
To make our choice in $R$, we multiply the weights of the tuples by $\factor=3$.
Then, we find that the index $k$ we are looking for falls into the range of $(a_2,c_3)$ because $1\cdot 3 \le 4 < 2\cdot 3$.
Next, $S$ is resolved, $k=4-1\cdot 3=1$, and $\factor=1$.
As $1\cdot 1 \le 1 < 2\cdot 1$, the tuple $(b_1,d_2)$ is selected.
Overall, answer number $12$ (the $13^{\textrm{th}}$ answer) is $(a_2,b_1,c_3,d_2)$.
\end{example}

\begin{lemma}\label{lemma:alg}
Let $Q$ be a full acyclic CQ, and $\lex=\angs{v_1, \ldots, v_f}$ be a lexicographic order.
If there is a layered join tree for $Q$ with respect to $\lex$, then direct access 
is possible 
in $\Comp{n \log n}{\log n}$.
\end{lemma}
\begin{proof}
The correctness of \Cref{alg:lex-random-access} follows from the discussion above.
For the time complexity, note that it
contains a constant number of operations (assuming the number of attributes $f$ is fixed). 
Line~\ref{line:select-tuple} can be done in logarithmic time using binary search, 
while all other operations only require constant time in the RAM model. 
Thus, we obtain direct access in logarithmic time per answer after the quasilinear preprocessing (dominated by sorting).
\end{proof}

\begin{figure}[t]
\centering
\includegraphics[scale=0.5]{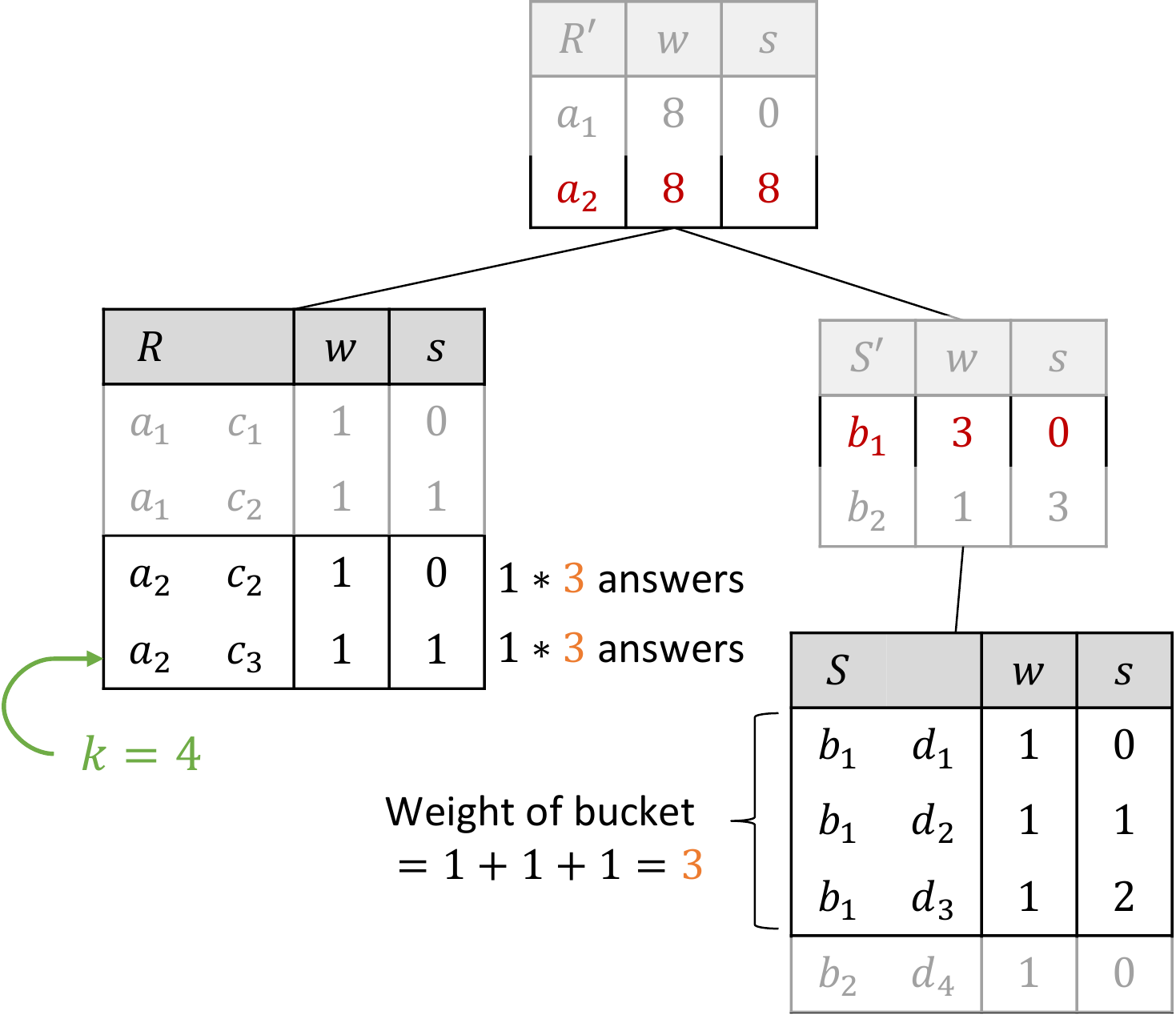}
\caption{\Cref{ex:lex_access}: Illustration of an iteration of the access phase where layer $3$ corresponding to $R$ is resolved.}
\label{fig:lex_access}
\end{figure}

\begin{remark}[Inverted access]
A straightforward adaptation of \Cref{alg:lex-random-access} can be used to achieve \e{inverted access}: 
given a query result as the input, we return its index according to the lexicographic order.
\Cref{alg:lex-inverted-access} is almost the same as \Cref{alg:lex-random-access} 
except that the choices in each iteration are made according to the given answer 
and the corresponding index is constructed (instead of the opposite). 
The algorithm runs in constant time per answer since every operation can be done within that time (unlike \Cref{alg:lex-random-access}, there is no need for binary search here).

Another adaptation of \Cref{alg:lex-inverted-access} can give us a form of inverted access for the cases when the given answer does not exist. 
That is, instead of returning ``not-an-answer'', we want to return the next answer in the lexicographic order.
The first time a tuple $\tuple$ is not found in Line \ref{line:tuple-not-found} of \Cref{alg:lex-inverted-access}, 
we select the first tuple in the bucket that is larger than $t$, 
and in all following iterations, we always select the first tuple in the bucket. 
If there is no tuple larger than $t$,
we revert the previous iteration and select the next tuple there (compared to what we selected before). 
If no such tuple exists, we again revert the previous iteration and so on. 
If there are no previous iterations, we were asked to access a tuple larger than the last answer, so we return an appropriate message. 
The algorithm described here takes logarithmic time, as we can use binary search to find the tuple following our target tuple in each bucket.
\end{remark}

\begin{algorithm}[t]
  $k = 0$\;
  $\bucket[1] = \root$\;
  $\factor = \w(\root)$\;
  \For{i=1,\ldots,f}{
  $\factor =\factor / \w(\bucket[i])$\;
  select $\tuple\in \bucket[i]$ agreeing with the answer\;
  \If{no such $\tuple$ exists}{ \label{line:tuple-not-found}
  \Return ``not-an-answer''
  }
  $k = k + \rng(\tuple)\cdot\factor$\;
      \For{child $V$ of layer $i$}{
          get the bucket $b\in V$ agreeing with the answer\;
          $\bucket[\layer(V)]=b$\;
          $\factor =\factor\cdot\w(b)$\;
    }
  }
  \Return $k$\;
\caption{Lexicographic Inverted-Access}
\label{alg:lex-inverted-access}
\end{algorithm}

\subsection{Finding Layered Join Trees}\label{sec:trio-to-layered}

We now have an algorithm that can be applied whenever we have a layered join tree. 
We next show that the existence of such a join tree relies on the disruptive trio condition we introduced earlier.
In particular, if no disruptive trio exists,
we are able to construct a layered join tree for full acyclic CQs.

\begin{lemma}\label{lemma:trio-to-layered}
Let $Q$ be a full acyclic CQ, and $\lex$ be a lexicographic order. 
If $Q$ does not have a disruptive trio with respect to $\lex$, 
then there is a layered join tree for $Q$ with respect to $\lex$.
\end{lemma}
\begin{proof}
We show by induction on $i$ that there exists a layered join tree for the hypergraph containing the hyperedges $\{V\cap\{v_1,\ldots,v_i\}\mid V\in\atoms(Q)\}$ with respect to the prefix of $\lex$ containing its first $i$ elements.
The induction base is the tree that contains the node $\{v_1\}$ and no edges.

In the inductive step,
we assume a layered join tree with $i-1$ layers for $\{V\cap\{v_1,\ldots,v_{i-1}\}\mid V\in\atoms(Q)\}$, 
and we build a layer on top of it.
Denote by $\calV$ the sets of $\{V\cap\{v_1,\ldots,v_i\}\mid V\in\atoms(Q)\}$ that contain $v_i$ (these are the sets that need to be included in the new layer).
First note that $\calV$ is acyclic. Indeed, by the running intersection property, the join tree for $\calH(Q)$ has a subtree with all the nodes that contain $v_i$. 
By taking this subtree and projecting out all variables that occur after $v_i$ in $\lex$, we get a 
join tree for 
an inclusion equivalent hypergraph to $\calV$,
and its existence proves that $\calV$ is acyclic.

We next claim that some set in $\calV$ contains all the others; that is, there exists $V_m\in\calV$ such that for all $V\in\calV$, we have that $V\subseteq V_m$.
Consider a join tree for $\calV$.
Every variable of $\calV$ defines a subtree induced by the nodes that contain this variable.
If two variables are neighbors, their subtrees share a node.
It is known that every collection of subtrees of a tree satisfies the \e{Helly property}~\cite{GOLUMBIC198081}: if every two subtrees share a node, then some node is shared by all subtrees.
In particular, since $\calV$ is acyclic, if every two variables of $\calV$ are neighbors, then some element of $\calV$ contains all variables that appear in (elements of) $\calV$.
Thus, if, by way of contradiction, there is no such $V_m$, there exist two non-neighboring variables $v_a$ and $v_b$ that appear in (elements of) $\calV$. Since $v_i$ appears in all elements of $\calV$, this means that there exist $V_a,V_b\in\calV$ with $\{v_a,v_i\}\subseteq V_a$ and $\{v_b,v_i\}\subseteq V_b$.
Since $v_a$ and $v_b$ are not neighbors, these three variables are a disruptive trio with respect to $\lex$: $v_a$ and $v_b$ are both neighbors of the later variable $v_i$.
The existence of a disruptive trio contradicts the assumption of the lemma we are proving, and so we conclude that there is $V_m\in\calV$ such that for all $V\in\calV$, we have that $V\subseteq V_m$.

With $V_m$ at hand, we can now add the additional layer to the tree given by the inductive hypothesis.
By the inductive hypothesis, 
the layered join tree with $i-1$ layers contains the hyperedge $V_m\cap\{v_1,\ldots,v_{i-1}\} = V_m\setminus\{v_i\}$.
We insert $V_m$ with an edge to the node containing $V_m\setminus\{v_i\}$.
This results in the join tree we need:
(1) the hyperedges $\{V\cap\{v_1,\ldots,v_i\}\mid V\in\atoms(Q)\}$ are all contained in nodes, since the ones that do not appear in the tree from the inductive hypothesis are contained in the new node;
(2) it is a tree since we add one leaf to an existing tree; and
(3) the running intersection property holds since the added node is connected to all of its variables that already appear in the tree.
\end{proof}

Lemmas~\ref{lemma:alg} and~\ref{lemma:trio-to-layered} give a direct-access algorithm for 
full acyclic CQs and lexicographic orders without disruptive trios.

\subsection{Supporting Projection}\label{sec:CQ-to-full}

Next, we show how to support CQs that have projections. 
A free-connex CQ can be efficiently reduced to a full acyclic CQ using \Cref{prop:reduce-to-subtree}. 
We next show that the resulting CQ contains no disruptive trio if the original CQ does not.

\begin{lemma}\label{lemma:CQ-to-full}
Given a database instance $I$, a free-connex CQ $Q$, and a lexicographic order $\lex$ with no disruptive trio with respect to $\lex$, 
we can compute 
in linear time a database instance $I'$
and a full acyclic CQ $Q'$ with no disruptive trio with respect to $\lex$ such that
$Q'(I') = Q(I)$,
$|I'| \leq |I|$,
and 
$Q'$ does not depend on $I$ or $I'$.
\end{lemma}
\begin{proof}
Let $Q$ be a free-connex CQ, and let $T$ be an ext-$\free(Q)$-connex tree for $Q$ where $T'$ is the subtree of $T$ that contains exactly the free variables.

First, we claim that two free variables are neighbors in $T$ iff they are neighbors in $T'$.
The ``if'' direction is immediate since $T'$ is contained in $T$. We show the other direction.
Let $u$ and $v$ be free variables of $Q$ that are neighbors in $T$. That is, there is a node $V_T$ in $T$ that contains them both.
Consider the unique path from $V$ to any node in $T'$ such that only the last node on the path, which we denote $V_{T'}$, is in $T'$. Since both variables appear in $T'$ and in $V$, by the running intersection property, both variables appear in $V_{T'}$.
Thus, $u$ and $v$ are also neighbors in $T'$.

Since the definition of disruptive trios depends only on neighboring pairs of free variables, an immediate consequence of the claim from the previous paragraph is that there is a disruptive trio in $T$ iff there is a disruptive trio in $T'$.
Next, we can simply use \Cref{prop:reduce-to-subtree} to reduce $Q$ to the full acyclic CQ where the atoms are exactly the nodes of $T'$.
\end{proof}

By combining \Cref{lemma:alg,lemma:trio-to-layered,lemma:CQ-to-full},
we conclude an efficient algorithm for free-connex CQs and orders with no disruptive trios.
The next lemma summarizes our results so far.

\begin{lemma}\label{lemma:easy-CQs}
Let $Q$ be a free-connex CQ, and $\lex$ be a lexicographic order. 
If $Q$ does not have a disruptive trio with respect to $\lex$, 
direct access by $\lex$ is possible 
in $\Comp{n \log n}{\log n}$.
\end{lemma}

\subsection{Lower Bound for Conjunctive Queries}\label{sec:dichotomy}

Next, we show that our algorithm supports all tractable cases (for self-join-free CQs); we prove that all unsupported cases are intractable.
We base our hardness results on the known hardness of enumeration for non-free-connex CQs~\cite{bagan2008computing,bb:thesis}
through a reduction that uses direct access to enumerate the answers 
projected on a prefix of the variables.

\begin{lemma}\label{lemma:enum-prefix-using-access}
Let $Q$ be a self-join-free CQ, $\lex$ be a lexicographic order,
and $Q'$ be the same as $Q$ but with free variables $L'$ for some prefix $L'$ of $L$.
If direct access for $Q$ by $L$ is possible in $\Comp{p(n)}{f(n)}$ for some functions $p, f$, then enumeration of the answers to $Q'$ is possible in $\Comp{p(n)}{f(n) \log n}$. 
\end{lemma}
\begin{proof}
We show how to enumerate the unique assignments of the free variables of $Q'$ 
given the direct access algorithm for $Q$.
First we perform the preprocessing step in $\bigO(p(n))$.
Then, we perform the following starting with $i=0$ and until there are no more answers.
We access the answer at index $i$ and print its assignment to the variables $L'$. 
Then, we set $i$ to be the index of the next answer which assigns $L'$ to different values and repeat.
Finding the next index can be done with a logarithmic number of direct access calls using binary search.
\end{proof}

We now exploit that for CQs with disruptive trios, 
we can always find a prefix that is not connex.
Therefore, enumerating the query answers projected on that prefix via direct access
leads to the enumeration of a non-free-connex CQ, where existing lower bounds apply.

\begin{lemma}\label{lemma:hardness}
Let $Q$ be a self-join-free acyclic CQ, and $\lex$ be a lexicographic order.
If $Q$ has a disruptive trio with respect to $\lex$, 
then direct access by $\lex$ is not possible
in $\Comp{n \polylog n}{\polylog n}$,
assuming \sparseBMM{}.\footnote{In fact, this lemma holds also for cyclic CQs, as it can be shown that Boolean matrix multiplication can be encoded in any CQ that contains a free-path regardless of its acyclicity. However, this is not formally stated in previous work, and we prefer not to complicate the proof with the technical details of the reduction. 
We chose here to limit the statement to acyclic CQs as cyclic CQs are already known to be hard if we assume \hyperclique{}. 
The direct reduction that applies also to cyclic CQs can be found in the conference version of this article \cite{carmeli21direct}.}
\end{lemma}
\begin{proof}
Let $v_1,v_2,v_3$ be a disruptive trio in $\lex$.
We take $L'$ to be the prefix of $L$ that ends in $v_2$. 
Then, $v_1,v_3,v_2$ is an $L'$-path or in other words,
the hypergraph of $Q$ is not $L'$-connex. 
Now, we define a new CQ $Q'$ so that it has the same body as $Q$ 
but its free variables are $L'$.
Thus, $Q'$ is acyclic but not free-connex.
Assuming that direct access for $Q$ is possible in $\Comp{n \polylog n}{\polylog n}$,
we use \Cref{lemma:enum-prefix-using-access} to enumerate the answers of $Q'$ in $\Comp{n \polylog n}{\polylog n}$,
which is known to contradict \sparseBMM{}~\cite{bagan2008computing} .
\end{proof}

By combining 
\Cref{lemma:easy-CQs} and \Cref{lemma:hardness} together with the known hardness results for non-free-connex CQs (\Cref{theorem:known-access}), we prove the dichotomy given in \Cref{thm:lex-dicotomy}: 
direct access by a lexicographic order for a self-join-free CQ is possible with quasilinear preprocessing and polylogarithmic time per answer
if and only if the query is free-connex and does not have a disruptive trio with respect to the required order.
\section{Direct Access by Partial Lexicographic Orders}
\label{sec:partial}

We now investigate the case where the desired lexicographic order is \e{partial}, i.e., it contains only some of the free variables.
This means that there is no particular order requirement for the rest of the variables.
One way to achieve direct access to a partial order is to complete it into a full lexicographic order and then leverage the results of the previous section.
If such completion is impossible, we have to consider cases where tie-breaking between the non-ordered variables is done in an arbitrary way.
However, we will show in this section that the tractable partial orders 
are precisely those 
that can be completed into a full lexicographic order.
In particular, we will prove the following dichotomy which also gives an easy-to-detect criterion for the tractability of direct access.

\begin{theorem}\label{thm:partial-dichotomy}
Let $Q$ be a CQ and $\lex$ be a partial lexicographic order.
\begin{itemize}
    \item If $Q$ is free-connex and $\lex$-connex and does not have a disruptive trio with respect to $\lex$, 
    then direct access by $\lex$ is possible 
    in $\Comp{n \log n}{\log n}$.
    \item Otherwise, if $Q$ is also self-join-free, then direct access by $\lex$ is not possible 
    in $\Comp{n \polylog n}{\polylog n}$,
    assuming \sparseBMM{} and \hyperclique{}.
\end{itemize}
\end{theorem}

\begin{example}
Consider the CQ $Q \datarule R(x,y),S(y,z)$.
If the free variables are exactly $x$ and $z$, then the query is not free-connex, and so it is intractable.
Next assume that all variables are free.
If $\lex=\angs{x,z}$, then the query is not $\lex$-connex, and so it is intractable.
If $\lex=\angs{x,z,y}$, then $x,z,y$ is a disruptive trio, 
thus the query is intractable.
However, if $\lex=\angs{x,y,z}$ or $\lex=\angs{z,y}$, then the query is free-connex, $\lex$-connex and has no disruptive trio, so it is tractable.
\end{example}

\subsection{Tractable Cases}

For the positive side, we can solve our problem efficiently if the CQ is free-connex and there is a completion of the lexicographic order to all free variables with no disruptive trio. 
\Cref{lemma:partial-completion} identifies these cases with a connexity criterion. 
To prove it, we first need a 
way to combine two different connexity properties.
The proof of the following proposition uses ideas from a proof of the characterization of free-connex CQs 
in terms of the acyclicity of the hypergraph obtained by including a hyperedge with the free variables~\cite{berkholz2020tutorial}.

\begin{proposition}\label{prop:two-connex}
If a CQ $Q$ is both $\lex_1$-connex and $\lex_2$-connex where $\lex_2\subseteq \lex_1$, 
then there exists a join tree $T$ of an inclusive extension of $Q$ with a subtree $T_1$ containing exactly the variables $\lex_1$ and a subtree $T_2$ of $T_1$ contains exactly the variables $\lex_2$.
\end{proposition}
\begin{proof}
We describe a construction of the required tree. \Cref{fig:inclusive-connexity} demonstrates our construction.
We use two different characterizations of connexity.
Since $Q$ is $\lex_2$-connex, it has an ext-$\lex_2$-connex tree $T_2$.
Since $Q$ is $\lex_1$-connex, there is a join-tree $T_1$ for the atoms of $Q$ and its head.
Let $T_2[\lex_1]$ be $T_2$ where the variables that are not in $\lex_1$ are deleted from all nodes. That is, for every node $V \in T_2$, its variables are replaced with $\var(V) \cap \lex_1$.
Denote by $\calV$ all neighbors of the head in $T_1$, 
and denote by $T_1^-$ the graph $T_1$ after the deletion of the head node.
Taking both $T_2[\lex_1]$ and $T_1^-$ and connecting every node $V_1 \in \calV$ with a node $V_2$ of $T_2[\lex_1]$ such that $\var(V_1) \cap \lex_1 = \var(V_2)$ gives us the tree we want.
Such a node exists in $T_2[\lex_1]$ since every node of $T_1^-$ represents an atom of $Q$, and every atom of $Q$ is contained in some node of $T_2$.
The subtree $T_2[\lex_1]$ contains exactly $V_1$, and since this subtree comes from an ext-$\lex_2$-connex tree, it has a subtree containing exactly $\lex_1$.
It is easy to verify that the result is a tree, and we can show that the running intersection property holds in the united graph since it holds for $T_1$ and $T_2$.
\end{proof}

\begin{figure}[t]
  \input{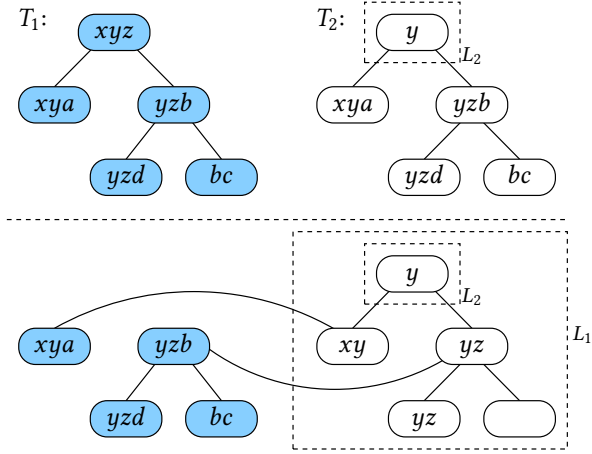}
  \caption{Example for the construction from \Cref{prop:two-connex} for the CQ $Q(x,y,z)\datarule R_1(x,y,a),R_2(y,z,b),R_3(b,c),R_4(y,z,d)$ with $L_1=\{x,y,z\}$ and $L_2=\{y\}$.}
  \label{fig:inclusive-connexity}
\end{figure}

We are now in a position to show the following:

\begin{lemma}\label{lemma:partial-completion}
Let $Q$ be a CQ and $\lex$ be a partial lexicographic order. 
If $Q$ is free-connex and $\lex$-connex and does not have a disruptive trio with respect to $\lex$, 
then there is an ordering $\lex^+$ of $\free(Q)$ that starts with $\lex$ such that $Q$ has no disruptive trio with respect to $\lex^+$.
\end{lemma}
\begin{proof}
According to \Cref{prop:two-connex},
there is a join tree $T$ 
(of an inclusive extension of $Q$)
with a subtree $T_{\free}$ containing exactly the free variables, 
and a subtree $T_\lex$ of $T_{\free}$ containing exactly the $\lex$ variables.
We assume that $T_\lex$ contains at least one node; otherwise (this can only happen in case $\lex$ is empty), we can introduce a node with no variables to all of $T$, $T_{\free}$ and $T_\lex$ and connect it to any one node of $T_{\free}$.
We describe a process of extending $\lex$ while traversing $T_{\free}$.
Consider the nodes of $T_\lex$ as handled, and initialize $\lex^+=\lex$.
Then, repeatedly handle a neighbor of a handled node until all nodes are handled. 
When handling a node, append to $\lex^+$ all of its variables that are not already there.
We prove by induction that $Q$ has no disruptive trio w.r.t any prefix of $\lex^+$.
The base case is guaranteed by the premises of this lemma since $\lex$ (hence all of its prefixes) has no disruptive trio.

Let $v_\varidx$ be a new variable added to a prefix $v_1,\ldots,v_{\varidx-1}$ of $\lex^+$.
Let $T^+$ be the subtree of $T_{\free}$ with the handled nodes when adding $v_\varidx$ to $\lex^+$ and let $V\not\in T^+$ be the node being handled.
Note that, since $v_\varidx$ is being added, $v_\varidx \in V$ but $v_\varidx$ is not in any node of $T^+$.

We first claim that every neighbor $v_i$ of $v_\varidx$ with $i<\varidx$ is in $V$.
Our arguments are illustrated in \Cref{fig:lex-completion}.
Since $v_i$ and $v_\varidx$ are neighbors, they appear together in a node $V_{i,\varidx}$ outside of $T^+$. 
Let $V_i$ be a node in $T^+$ containing $v_i$ (such a node exists since $v_i$ appears before $v_\varidx$ in $\lex^+$). 
Consider the path from $V_{i,\varidx}$ to $V_{i}$. 
Let $V_{\ell}$ be the last node of this path not in $T^+$. If $V_{\ell} \neq V$,
the path between $V_{\ell}$ and $V$ goes only through nodes of $T^+$ (except for the end-points).
Thus, concatenating the path from $V_{i,\varidx}$ to $V_{\ell}$ with the path from $V_{\ell}$ to $V$ results in a simple path. 
By the running intersection property, all nodes on this path contain $v_\varidx$. In particular, the node following $V_{\ell}$  contains $v_\varidx$ in contradiction to the fact that $v_\varidx$ does not appear in $T^+$. Therefore, $V_{\ell}= V$. By the running intersection property, since $V$ is on the path between $V_i$ and $V_{i,\varidx}$, we have that $V$ contains $v_i$.

\begin{figure}[t]
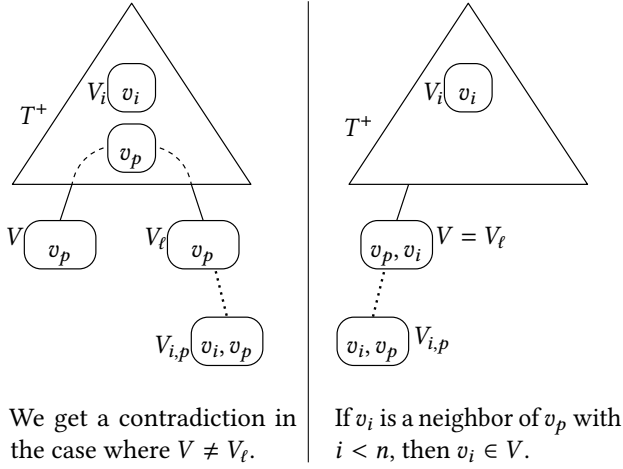

  \parbox{1.5in}{
    \input{partial2.pspdftex}
    \vskip1em
    
    We get a contradiction in the case where $V\neq V_\ell$.
  }
  \,\vrule\quad
  \parbox{1.5in}{
    \input{partial1.pspdftex}
    \vskip1em
    
    If $v_i$ is a neighbor of $v_p$ with $i<n$, then $v_i\in V$.
    }
  \caption{The induction step in \Cref{lemma:partial-completion}}
  \label{fig:lex-completion}
\end{figure}

We now prove the induction step. We know by the inductive hypothesis that $v_1,\ldots,v_{\varidx-1}$ have no disruptive trio. Assume by way of contradiction that appending $v_\varidx$ introduces a disruptive trio. 
Then, there are two variables $v_i,v_j$ with $i<j<\varidx$ such that $v_i,v_\varidx$ are neighbors, $v_j,v_\varidx$ are neighbors, but $v_i,v_j$ are not neighbors.
As we proved, since $v_i$ and $v_j$ are neighbors of $v_\varidx$ preceding it, we have that all three of them appear in the handled node $V$. 
This is a contradiction to the fact that $v_i$ and $v_j$ are not neighbors.
\end{proof}

The positive side of \Cref{thm:partial-dichotomy} is obtained by combining \Cref{lemma:partial-completion} with \Cref{thm:lex-dicotomy}.

\subsection{Intractable Cases}\label{sec:partial-intractable}

For the negative part, we prove a generalization of \Cref{lemma:hardness}.
Recall that according to \Cref{lemma:enum-prefix-using-access}, 
we can use lexicographic direct access to enumerate the answers to a CQ with a prefix of the ordered free variables. 
Similarly to \Cref{sec:dichotomy}, our goal is to find a 
``bad'' prefix that does not allow efficient enumeration.
For non-$L$-connex CQs, this is easy since $L$ itself is such a prefix.

\begin{lemma}\label{lemma:enum-nonfc-using-access}
Let $Q$ be an acyclic self-join free CQ and $L$ be a partial lexicographic order. 
If $Q$ has a disruptive trio or $Q$ is not $L$-connex, then there exists a self-join-free acyclic non-free-connex CQ $Q'$ such that: if direct access for $Q$ is possible in $\Comp{n \polylog n}{\polylog n}$, then enumeration for $Q'$ is possible in $\Comp{n \polylog n}{\polylog n}$.
\end{lemma}
\begin{proof}
If $Q$ is not $L$-connex, we use \Cref{lemma:enum-prefix-using-access} with $L'=L$. If $L$ has a disruptive trio $v_1,v_2,v_3$, we take $L'$ to be the prefix of $L$ that ends in $v_2$. 
Then, $v_1,v_3,v_2$ is an $L'$-path, meaning that the body of $Q$ is not $L'$-connex.
Thus, we can use \Cref{lemma:enum-prefix-using-access} in that case too.
\end{proof}

It is known that, assuming $\sparseBMM{}$, self-join-free non-free-connex CQs cannot be answered with polylogarithmic time per answer after quasilinear preprocessing time. Thus, we conclude from \Cref{lemma:enum-nonfc-using-access} that self-join-free acyclic CQs with disruptive trios or that are not $L$-connex do not have partial lexicographic direct access within these time bounds either.
The case that $Q$ is cyclic is hard since even finding any answer for cyclic CQs is not possible efficiently assuming \hyperclique{}.

\section{Direct Access by Sum of Weights}
\label{sec:sum_direct}

We now consider direct access for the more general orderings based on SUM
(the \e{sum} of free-variable weights). 
As with lexicographic orderings,
we are able to exhaustively classify tractability for the self-join-free CQs, 
even those with projections.
We will show that direct access for SUM is significantly harder and
tractable only for a small class of queries.

\subsection{Overview of Results}

The main result of this section is a dichotomy for direct access by SUM orders:

\begin{theorem}[Direct Access by SUM]
\label{th:ra_sum_dichotomy}
Let $Q$ be a CQ.
\begin{itemize}
    \item If $Q$ is acyclic and an atom of $Q$ contains all the free variables,
    then direct access by SUM is possible 
    in $\Comp{n \log n}{1}$.
    \item Otherwise, if $Q$ is also self-join-free, 
     direct access by SUM is not possible 
     in $\Comp{n \polylog n}{\polylog n}$,
     assuming \ThreeSUM{} and \hyperclique{}.
\end{itemize}
\end{theorem}

For the positive part of the above theorem, we will see that we are able to materialize the query answers and keep them in
a sorted array that supports direct access in constant time.
The proof of the negative part requires the query answers to express
certain combinations of weights. If the query contains \e{independent} free variables,
then its answers may contain all possible combinations of their corresponding attribute
weights. 
We will thus rely on
this independence measure 
to identify hard cases.

\begin{definition}[Independent free variables]
	\label{def:indVar}
A set of vertices $V_i \subseteq V$ of a hypergraph $\calH(V, E)$ is called 
independent iff no pair of these vertices appears in the same hyperedge, i.e.,
$|V_i \cap e| \leq 1$ for all $e \in E$.
For a CQ $Q$, we denote by $\freeind(Q)$
the maximum number of variables among $\free(Q)$ that are independent in $\calH(Q)$.
\end{definition}

Intuitively, we can construct a database instance 
where each independent free variable is assigned to 
$n$ different domain values with $n$ different weights.
By appropriately choosing the assignment of the other variables,
all possible $n^{\freeind(Q)}$ combinations of these weights will appear in the query answers.
Providing direct access then implies 
that we can retrieve these sums in ranked order.
We later use this to show that direct access on certain CQs allows us to solve \ThreeSUM{} efficiently.

\begin{example}
\label{ex:3sumConstruction}
For $Q(x, y, z) \datarule R(x,y), S(y,z), T(z,u)$, 
we have $\freeind(Q) = 2$, 
namely for variables $\{x, z\}$. 
Let the binary relation $R$ be $[1, n] \times \{ 0 \}$, i.e.,
the cross product between the set of values from $1$ to $n$ with the single value $0$.
If we also set
$S = \{ 0 \} \times [1, n]$ and
$T = [1, n] \times \{ 0 \}$,
then the query answers are the $n^2$ assignments of $(x, y, z)$ to $[1, n] \times [1, n] \times \{ 0 \}$.
The $n$ values of $x$ and $z$ can be respectively assigned to any real-valued weights
such that direct access on $Q$ retrieves their $i^\textrm{th}$ sum in ranked order.
\end{example}

Our independence measure $\freeind(Q)$ is related to the classification of \Cref{th:ra_sum_dichotomy} in the following way:

\begin{lemma}
\label{lem:freeind}
For an acyclic CQ $Q$, an atom contains all the free variables iff $\freeind(Q) \leq 1$.
\end{lemma}

\begin{proof}

The ``only if'' part of
$\freeind(Q) > 1$ follows immediately from \cref{def:indVar}.

For $\freeind(Q) = 1$ and acyclic query $Q$, 
we prove that there is an atom $R_f(\mathbf{X}_f)$ which contains all the free variables.
First note that for $|\free(Q)| = 1$ this is trivially true. For $|\free(Q)| > 1$,
let $V$ 
be a node in the join tree 
(corresponding to some atom of $Q$)
that contains the maximum number of free variables
and assume for the sake of contradiction that there exists a free variable $y$
with $y \notin V$. We use $\calV_y$ to denote the set of nodes in the join tree that
contain variable $y$; thus $V \notin \calV_y$. 
From $Q$ being acyclic follows that
the nodes in $\calV_y$ form a connected graph and there exists a node $V'$
that lies on every path from $V$ to a node in $\calV_y$.
Since $\freeind(Q) = 1$, each variable $x \in V$ must appear together with
$y$ in some query atom, implying that $x$ appears in some node $V'' \in \calV_y$.
From that and the running intersection property follows that $x$ must
also appear in $V'$ since $V'$ lies on the path from $V$ to any such $V''$.
Hence $V'$ contains $y$ and all the $V$ variables,
violating the maximality assumption for $V$.

For $\freeind(Q) = 0$, $Q$ is a Boolean query
and any atom trivially contains the empty set.
\end{proof}

Therefore, the dichotomy of \Cref{th:ra_sum_dichotomy} can equivalently be stated using $\freeind(Q) \leq 1$ as a criterion.
We chose to use the other criterion (all free variables contained in one atom) in the statement of our theorem statement as it is more straightforward to check.
In the next section, we proceed to prove our theorem by showing intractability for all queries with $\freeind(Q) > 1$
and a straight-forward algorithm for $\freeind(Q) \leq 1$.

\subsection{Proofs}

\begin{figure}[t]
\centering
\begin{tabular}[t]{ c | c  c c }
Query condition			& Direct access 		& Complexity 								& Reason \\
\hline                                  	
acyclic $\freeind(Q) = 1$		& possible in			& $\Comp{n \log n}{1}$ 						& \Cref{lem:rra_1}		\\			
acyclic $\freeind(Q) = 2$		& not possible in		& $\Comp{n^{2-\epsilon}}{n^{1-\epsilon}}$	& \ThreeSUM{}		\\
acyclic $\freeind(Q) \geq 3$	& not possible in 		& $\Comp{n^{2-\epsilon}}{n^{2-\epsilon}}$  	& \ThreeSUM{}		\\
cyclic							& not possible in 		& $\Comp{n \polylog n}{\polylog n}$			& \hyperclique{}
\end{tabular}			
\caption{Possibility of direct access by sum of weights for acyclic self-join-free conjunctive queries.}
\label{fig:directAccessSum}
\end{figure}

For the hardness results, we rely mainly on the \ThreeSUM{} hypothesis.  
To more easily relate our direct-access problem to \ThreeSUM{},
which asks for the existence of a particular sum of weights,
it is useful to define an auxiliary problem:

\begin{definition}[weight lookup]
Given a CQ $Q$ and a weight function $w$ over its possible answers, 
\emph{weight lookup} takes as an input a database $I$ and
$\lambda \in \R$, 
and returns the first index of a query answer $q \in Q(I)$
with $w(q) = \lambda$ in the array of answers sorted by $w$
or ``none'' if no such answer exists.
\end{definition}

The following lemma associates direct access with weight lookup via binary search on the query answers:

\begin{lemma}\label{lem:inverted_sum}
For a CQ $Q$, if the $k^{\textrm{th}}$ query answer ordered by a weight function $w$
can be
directly accessed in $\bigO(T_d(n))$ time for every $k$, 
then weight lookup for $Q$ and $w$ can be performed
in $\bigO(T_d(n) \log n)$.
\end{lemma}
\begin{proof}
We use binary search on the sorted array of query answers.
Each direct access returns a query answer whose weight can be computed in $\bigO(1)$.
Thus, in a logarithmic number of accesses we can find the first occurrence of the
desired weight.   Since the number of answers is polynomial in $n$, the number
of accesses is $\bigO(\log n)$ and each one takes $\bigO(T_d(n))$ time.
\end{proof}

\Cref{lem:inverted_sum} implies that whenever we are able to support efficient
direct access on the sorted array of query answers, weight lookup increases
time complexity only by a logarithmic factor, i.e., it is also efficient.
The main idea behind our reductions is that via weight lookups on a CQ with an
appropriately constructed database,  we can decide the existence of a zero-sum
triplet over three distinct sets of numbers, thus hardness follows from \ThreeSUM{}.
First, we consider the case of three independent variables that are free. 
These three variables are able to simulate a three-way Cartesian product in the query
answers. This allows us to directly encode the \ThreeSUM{} triplets using
attribute weights, obtaining a lower bound for direct access.

\begin{lemma}\label{lem:3sum_to_rra_geq3}
If a CQ $Q$ is self-join-free and $\freeind(Q) \geq 3$, 
then direct access by SUM is not possible
in $\Comp{n^{2-\epsilon}}{n^{2-\epsilon}}$
for any $\epsilon > 0$ assuming \ThreeSUM{}.
\end{lemma}

\begin{proof}
Assume for the sake of contradiction that the lemma does not hold. We show that
this would imply an $O(n^{2-\epsilon})$-time algorithm for \ThreeSUM{}.
To this end, consider an instance of \ThreeSUM{} with integer sets
$A$, $B$, and $C$ of size $n$, given as arrays.
We reduce \ThreeSUM{} to direct access over the appropriate query and input
instance by using a construction similar to \Cref{ex:3sumConstruction}.
Let $x$, $y$, and $z$ be free and independent variables of $Q$, 
which exist because $\freeind(Q) \geq 3$. 
We create a database instance where $x$, $y$, and $z$ take on each value in $[1, n]$,
while all the other attributes have value $0$. 
This ensures that $Q$
has exactly $n^3$ answers---one for each $(x,y,z)$ combination in $[1,n]^3$,
no matter the number of atoms and 
the variables they contain.
To see this, note that
since $x$, $y$, and $z$ are independent, 
no pair of them appears together in an atom.
Also, since $Q$ is self-join-free, 
each relation appears once in the query,
hence contains at most one of $x$, $y$, and $z$.
Thus each relation either contains $1$ tuple (if neither $x$, $y$, nor $z$ is present)
or $n$ tuples (if one of $x$, $y$, or $z$ is present).
No matter on which attributes these relations
are joined (including Cartesian products), the output result is always the ``same''
set $[1,n]^3 \times \{ 0 \}^f$ of size $n^3$, where $f$ is the number of free
variables other than $x$, $y$, and $z$.
(We use the term ``same'' loosely for the sake of simplicity. Clearly,
for different values of $f$ the query-result schema changes, e.g., consider
\cref{ex:3sumConstruction} with $z$ removed from the head. However, this only
affects the number of additional $0$s in each of the $n^3$ answer tuples, therefore
it does not impact our construction.)

For the reduction from \ThreeSUM{}, weights are assigned to the attribute values as
$w_x(i) = A[i]$, $w_y(i) = B[i]$, $w_z(i) = C[i]$, $i \in [1, n]$, and
$w_u(0) = 0$ for all other attributes $u$.
By our weight assignment, the weights of the answers are
$A[i] + B[j] + C[k]$, $i, j, k \in [1, n]$, and thus
in one-to-one correspondence with the possible value combinations in the \ThreeSUM{} problem.
We first perform the preprocessing for direct access in $O(n^{2-\epsilon})$, which
enables direct access to any position in the sorted array of query answers in
$O(n^{2-\epsilon})$. By \Cref{lem:inverted_sum}, weight lookup for a query result
with zero weight is possible in $O(n^{2-\epsilon} \log n)$.
Thus, we answer the original \ThreeSUM{} problem in $O(n^{2-\epsilon'})$
for any $0 < \epsilon' < \epsilon$, violating the \ThreeSUM{} hypothesis.
\end{proof}

For queries that do not have three independent free variables, we need a slightly different
construction. We show next that two variables are sufficient to encode partial
\ThreeSUM{} solutions (i.e., pairs of elements), enabling a full solution of \ThreeSUM{}
via weight lookups. This yields a weaker lower bound than \Cref{lem:3sum_to_rra_geq3},
but still is sufficient to prove intractability according to our yardstick.

\begin{lemma}
\label{lem:3sum_to_rra_2}
If a CQ $Q$ is self-join-free and $\freeind(Q) = 2$, 
then direct access by SUM is not possible
in $\Comp{n^{2-\epsilon}}{n^{1-\epsilon}}$
for any $\epsilon > 0$ assuming \ThreeSUM{}.
\end{lemma}
\begin{proof}
We show that a counterexample query would violate
the \ThreeSUM{} hypothesis.
Let $A$, $B$, and $C$ be three integer arrays of a \ThreeSUM{} instance of size $n$.
We construct a database instance with attribute weights like in the proof of
\Cref{lem:3sum_to_rra_geq3}, but now with only 2 free and independent variables
$x$ and $y$. Hence the weights of the $n^2$ query results are in one-to-one
correspondence with the corresponding sums $A[i]+B[j]$, $i, j \in [1,n]$.
We run the preprocessing phase for direct access in $O(n^{2-\epsilon})$,
which allows us to access the sorted array of query results in $O(n^{1-\epsilon})$.
For each value $C[k]$ in $C$, we perform a weight lookup on $Q$ for weight $-C[k]$,
which takes time $O(n^{1-\epsilon} \log n)$ (\Cref{lem:inverted_sum}). 
If that returns a valid index, then there exists a pair $(i, j)$ of $A$ and $B$ with sum 
$A[i] + B[j] = -C[k]$, which implies $A[i] + B[j] + C[k] = 0$; otherwise no such pair exists.
Since there are $n$ values in $C$, total time complexity is
$\bigO(n \cdot n^{1-\epsilon} \log n) = \bigO(n^{2-\epsilon} \log n)$.
This procedure solves \ThreeSUM{} in $O(n^{2-\epsilon'})$
for any $0 < \epsilon' < \epsilon$, violating the \ThreeSUM{} hypothesis.
\end{proof}

A special case of \Cref{lem:3sum_to_rra_2} is closely related to
the problem of selection in $X+Y$ \cite{johnson78xy},
where we want to access the $k^{\mathrm{th}}$ smallest sum of pairs between two sets $X$ and $Y$.
This is equivalent to accessing the answers to $Q_{XY}(x, y) \datarule R(x), S(y)$
by a SUM order.
It has been shown that if $X$ and $Y$ are given sorted, 
then selection (single access) is possible even in \emph{linear}
time~\cite{frederickson84selection,mirazaian85xy}.
Thus, for $Q_{XY}$ direct access by SUM is possible in $\Comp{n \log n}{n}$
if we sort the relations during the preprocessing phase.
Compared to our $\Comp{n^{2-\epsilon}}{1-\epsilon}$ lower bound (see also \Cref{fig:directAccessSum}),
notice that even though the preprocessing of this algorithm is lower (asymptotically), the access time is not sublinear ($epsilon=0$).

So far, we have covered all self-join-free CQs with $\freeind(Q) > 1$, which, by \Cref{lem:freeind}, proves the negative part of \Cref{th:ra_sum_dichotomy}.
Next, we show that the remaining acyclic CQs (those with $\freeind(Q) \leq 1$ or equivalently, an atom containing all the free variables) 
are tractable.
For these queries,
a single relation contains all the answers, so direct access
can easily be supported by reducing, projecting, and sorting that relation.

\begin{lemma}
\label{lem:rra_1}
If a CQ $Q$ is acyclic and an atom contains all the free variables,
then direct access by SUM is possible 
in $\Comp{n \log n}{1}$.
\end{lemma}

\begin{proof}
Since all free variables appear in one atom $R_f(\mathbf{X}_f)$, we can apply a linear-time 
semi-join reduction as in the Yannakakis algorithm~\cite{Yannakakis}
to remove the dangling tuples, 
and then compute the query answers
by projecting $R$ on the
free variables. Then, we sort the query answers by the sum of weights, which takes total
time $\bigO(n \log n)$ for preprocessing.
We maintain the sorted answers in an array,
which enables constant-time direct access
to individual answers in ranked order.
\end{proof}

We now combine these lemmas with the fact that Boolean self-join-free cyclic CQs cannot be answered in $\bigO(n \polylog n)$ time assuming \hyperclique{},
completing the proof of \Cref{th:ra_sum_dichotomy}.

\section{Selection by Lexicographic Orders}
\label{sec:lex_selection}

We next investigate the tractability of a simpler version of the problem:
When is \emph{selection}, i.e., direct access to a \emph{single} query answer,
possible in quasilinear time?
In this section, we answer this question for lexicographic orders and in \Cref{sec:sum_selection} we move to the case of SUM.
Unlike direct-access, we show that selection can be efficiently achieved for any lexicographic order, as long as the query is free-connex.
Our main result in this setting is summarized below:

\begin{theorem}[Selection by LEX]\label{thm:lex-selection-dicotomy}
Let $Q$ be a CQ and $\lex$ be a partial lexicographic order.
\begin{itemize}
    \item If $Q$ is free-connex, 
    then selection by $\lex$ is possible 
    in $\Comp{1}{n}$.
	
    \item Otherwise, if $Q$ is also self-join-free, then selection by $\lex$ is not possible 
    in $\Comp{1}{n \polylog n}$,
    assuming \seth{} and \hyperclique{}.
\end{itemize}
\end{theorem}

Our theorem shows that when we limit ourselves to the problem of selection,
the tractability of the problem depends only on the query structure and 
is independent of the lexicographic order.

\begin{example}
Recall that direct access by $L$ is intractable for $Q(v_1, v_2, v_3) \datarule R(v_1, v_3), S(v_3, v_2)$ 
with $L$ being the lexicographic order $\angs{v_1, v_2, v_3}$ 
or the partial lexicographic order $\angs{v_1, v_2}$.
The former contains a disruptive trio while the latter is not $L$-connex.
However, selection is tractable in both cases. 
Still, if we project out the middle variable $v_3$ and the head of the CQ is $Q(v_1, v_2)$, 
then the CQ is not free-connex and thus, selection becomes intractable for any lexicographic order.
\end{example}

For the negative part of \cref{thm:lex-selection-dicotomy},
we reduce the problem of selection to that of counting query answers.

\begin{lemma}\label{lem:lex-selection-to-counting}
For a CQ $Q$, if selection by some ranking function is possible in $\Comp{1}{f(n)}$ for a function $f$,
then counting the answers to the CQ $Q$ is possible in $\bigO(f(n) \log n)$.
\end{lemma}
\begin{proof}
We reduce the counting problem to the selection problem under any ranking function.
If the number of relations in the query is $\ell$, then an upper bound on the number of answers is $n^\ell$.
We use selection to determine whether any index contains an answer.
With binary search 
on the range of indices $[0, n^\ell)$,
we can find the smallest index that does not correspond to an answer.
This process requires only $\bigO(\log n^\ell) = \bigO(\log n)$ selections since $\ell$ is constant.
\end{proof}

We can now exploit
lower bounds based on \seth{}.
The proof does not rely on the properties of lexicographic orders and 
thus captures any possible ordering of the query answers.

\begin{lemma}\label{lem:lex-selection-intractability}
If a self-join-free CQ $Q$ is not free-connex, 
then selection by any ranking function is not possible 
in \Comp{1}{n \polylog n} assuming \seth{} and \hyperclique{}.
\end{lemma}

\begin{proof}
We use the fact that, assuming \seth{}, the answers to a self-join-free and acyclic non-free-connex CQ 
cannot be counted in $\bigO(n^{2-\epsilon})$ for any constant $\epsilon$ \cite{stefan}.
By \Cref{lem:lex-selection-to-counting},
if selection is possible in $\bigO(n \polylog n)$,
then we can also count the number of query answers in $\bigO(n \polylog n)$,
contradicting our hypothesis.
Cyclic CQs are covered by the hardness of Boolean self-join-free cyclic CQs based on \hyperclique{},
completing the proof.
\end{proof}

For the remainder of this section, we give a selection algorithm that together with 
\cref{lem:lex-selection-intractability}
completes the proof of \cref{thm:lex-selection-dicotomy}

\subsection{Lexicographic Selection Algorithm}
\label{sec:lex_selection_algorithm}

We first claim that, for any free variable in a free-connex CQ, we can efficiently compute the histogram of its assignments in the query answers.
This is essentially equivalent to a group-by query that groups the query answers 
based on a single variable, and then
counts how many answers fall within each group.

\begin{lemma}\label{lemma:count-histogram}
Let $Q$ be a free-connex CQ and $v \in \free(Q)$. 
Given an input database $I$, we can compute in linear time
how many answers in $Q(I)$ assign $c$ to $v$ for each value $c$
in the active domain of $v$.
\end{lemma}
\begin{proof}
Following \Cref{prop:reduce-to-subtree}, we can transform the problem to an equivalent problem with a full acyclic CQ $Q'$. 
We then take a join-tree for $Q'$, identify a node $V_p$ containing $v$, and introduce a new node $V_r$ as a neighbor of $V_p$. 
We associate $V_r$ with the single variable $v$, 
assign $V_r$ with a unary relation that contains the active domain of $v$, and set $V_r$ to be the root of the tree. 
Then, we follow the preprocessing explained in \Cref{sec:layered-alg} over this tree. 
By the end of this preprocessing, each tuple is given a weight that indicates the number of different answers that this tuple agrees with when only joining its subtree. 
Thus, the weights for $V_r$ will contain the desired values.
\end{proof}

This count guides our selection algorithm
as it iteratively chooses an assignment for the next variable in the lexicographic order.
Comparing the desired index with the count, it chooses an appropriate value for 
the next variable,
filters the remaining relations according to the chosen value,
and continues with the next variable.

\begin{lemma}
Let $Q$ be a free-connex CQ and $\lex$ be a partial lexicographic order.
Then, selection by $\lex$ is possible in $\Comp{1}{n}$.
\end{lemma}
\begin{proof}
Let $\angs{v_1,\ldots,v_m}$ be a completion of $\lex$ to a full lexicographic order,
and let $k$ be the index we want to access.
We perform the following starting with $i=1$.
Let $c_1,\ldots,c_m$ be the ordered values in the active domain of $v_i$
(the algorithm does not sort them because that would already take $\bigO(n \log n)$).
We use \Cref{lemma:count-histogram} to count, for each $c_r$, 
the number of answers that assign $c_r$ to $v_i$, denoted by $\w(c_r)$. 
Then, we find $j$ such that 
$\sum_{r=1}^{j-1}{\w(c_r)}\le k < \sum_{r=1}^{j}{\w(c_r)}$ and 
select the value $c_j$ for $v_i$.
This computation can be done in $\bigO(n)$ without sorting
if we use a weighted selection algorithm~\cite{johnson78xy}.
We proceed to filter all relations according to the $v_i = c_j$ assignment, update $k$ to $k-\sum_{r=1}^{j-1}{\w(c_r)}$, 
and continue iteratively with $i+1$.
In each iteration, the value for another variable is determined, where $\sum_{r=1}^{j-1}{\w(c_r)}$ answers contain a strictly smaller value for the variable, 
and the next iterations break the tie between the $\w(c_j)$ answers that have this value.
For the running time, each iteration takes linear time and we have a constant number of iterations (one iteration for every free variable).
\end{proof}

\section{Selection by Sum of Weights}
\label{sec:sum_selection}

We now move on to the problem of selection by SUM order.
Given that direct access by this order with quasilinear preprocessing
and polylogarithmic delay is possible only in very few cases,
it is a natural question to ask how the tractability landscape changes when considering the simpler task of \emph{selection}.

\subsection{Overview of Results}

We show that the simplifications move only a narrow class of queries
to the tractable side.
For example, the 2-path query $Q_2(x, y, z) \datarule R(x, y), S(y, z)$
is tractable for selection, even though it is not for direct access.
On the other hand, the 3-path query
$Q_3(x, y, z, u) \datarule R(x, y), S(y, z), T(z, u)$ remains intractable.
Given that $Q_2$ and $Q_3$ both have two free and independent variables,
a different criterion than that of \Cref{sec:sum_direct}
($\freeind(Q)$ or number of atoms containing the free variables) 
is needed for classification.
To this end, we use hypergraph $\calH_\free(Q)$.
Recall that it is the restriction of the query hypergraph $\calH(Q)$ to the free variables, 
i.e., all the other variables are removed.

\begin{definition}[Maximal Hyperedges]
For a hypergraph
$\calH = (V, E)$, we denote the \e{number of maximal hyperedges} w.r.t. containment
by $\mh(\calH )$, i.e.,
$\mh(\calH ) = | \{ e \in E \;|\; \nexists e' \in E: e \subset e' \}|$.
The number of maximal hyperedges of a query $Q$ is $\mh(Q) = \mh(\calH(Q))$ 
and the number of 
free-maximal hyperedges of $Q$ is $\mhfree(Q) = \mh(\calH_\free(Q))$.
\end{definition}

\begin{example}
For $Q(x, z, w) \datarule R(x, y), S(y, z), T(z, w), U(x)$, we have
$\mh(Q)=3$ because $U$ is contained in $R$ 
and $\mhfree(Q) = 2$ because after removing the existentially-quantified $y$,
the remainder of the $S$-hyperedge is contained in $T$.
\end{example}

\begin{remark}
For any CQ $Q$ we have $\freeind(Q) \leq \mhfree(Q)$. This follows from the fact
that each independent variable must appear in a maximal hyperedge and that
each hyperedge cannot contain more than 1 independent variable by definition.
Note also that the condition $\freeind(Q) \leq 1$ is equivalent to 
$\mhfree(Q) \leq 1$, 
giving us a third possible way to express the criterion of \Cref{th:ra_sum_dichotomy} for direct access. 
\end{remark}

We summarize the results of this section in the following theorem, which classifies CQs $Q$ based on $\mhfree(Q)$:

\begin{theorem}[Selection by SUM]\label{thm:sum-selection-dicotomy}
Let $Q$ be a CQ.
\begin{itemize}
    \item If $Q$ is free-connex and $\mhfree(Q) \leq 2$,
    then selection by SUM is possible in 
    $\Comp{1}{n \log n}$.
    \item Otherwise, if $Q$ is also self-join-free, 
    then selection by SUM is not possible in %
    $\Comp{1}{n \polylog n}$.
    assuming \ThreeSUM{}, \hyperclique{}, and \seth{}.
\end{itemize}
\end{theorem}

\begin{example}
For the query $Q_2(x, y, z) \datarule R(x, y), S(y, z)$ we have already shown in \Cref{sec:sum_direct} that direct access by SUM is intractable. 
However, given that it has two maximal hyperedges, 
only one access (or a constant number of them) is in fact possible in $\bigO(n \log n)$.
The situation does not change for $Q_3'(x, y, z) \datarule R(x, y), S(y, z), T(z, u)$
because the hyperedge of $T$ is contained in $S$ in the free-restricted hypergraph.
However, $Q_3(x, y, z, u) \datarule R(x, y), S(y, z), T(z, u)$ which keeps the variable $u$ in the answers 
is intractable for selection because now $T$ corresponds to a free-maximal hyperedge.
\end{example}

Before proving \Cref{thm:sum-selection-dicotomy}, we first introduce some necessary
concepts and prove a useful lemma.

\smallsection{Absorbed atoms and variables}
We say that an atom $e$ is \emph{absorbed} by an atom $e' \neq e$ if 
$V \subseteq V'$ where $V$ and $V'$ are their sets of variables respectively.
Additionally, we say that a variable $v$ is absorbed by a variable $u \neq v$ if 
(1) they appear in exactly the same atoms and 
(2) it is not the case that $v$ is free and $u$ is not free.
As evident from \Cref{thm:sum-selection-dicotomy}, 
adding to a query atoms or variables that are absorbed by existing ones does not affect the complexity of selection.
We prove this claim first and use it later in our analysis in order to treat queries that contain absorbed atoms or variables.

\begin{definition}[Maximal Contraction]
A query $Q'$ is a \emph{contraction} of $Q$ if we can obtain $Q'$ by iteratively removing 
absorbed atoms and variables, one at a time.
$Q^m$ is a \emph{maximal contraction} of $Q$ if it is a contraction and there is no contraction of $Q^m$.
\end{definition}

Note that the number of atoms of a maximal contraction $Q^m$ of $Q$ is $\mh(Q)$.

\begin{example}
Consider $Q(x, y, z) \datarule R(x, u, y), S(y), T(y, z), U(x, u, y)$.
Here, $S(y)$ is absorbed by $R(x, u, y)$ and $U(x, u, y)$, and the latter two absorb each other. 
Additionally, the free variable $x$ absorbs $u$ since these two variables appear together in $R$ and $U$.
Thus, a maximal contraction of $Q$ is $Q^m(x, y, z) \datarule R(x, y), T(y, z)$, which is unique up to renaming.
The number of maximal hyperdges of $Q$ is $\mh(Q) = 2$.
\end{example}

\begin{lemma}
\label{lem:absorbed}
Selection for a CQ $Q$ by SUM is possible in $\Comp{1}{T_S(n)}$ if
selection for a maximal contraction $Q^m$ of $Q$ by SUM is possible in $\Comp{1}{T_S(n)}$.
The converse is also true if $Q$ is self-join-free.
\end{lemma}
\begin{proof}
For the ``if'' direction, we use selection on $Q^m$ to solve selection on $Q$.
We can remove absorbed atoms from $Q$
after making sure that the tuples in the database satisfy those atoms.
Thus, to remove an atom $S(\mathbf{Y})$ which is absorbed by $R(\mathbf{X})$,
we filter the relation $R$ based on the tuples of $S$.
To remove a variable $v$ that is absorbed by $u$, in all relations
that contain both $u$ and $v$ we remove $v$ and replace the $u$-values by values
that represent the pair $(u,v)$ and assign to it the weight $w(u) + w(v)$.
(Note that we assign $w(z) = 0$ for all variables $z$ that are not free.)
After separating the packed variables, $Q^m$ over the modified database has the same answers as $Q$ 
over the original one
and the weights are preserved.

For the ``only if'' direction, 
we create an extended database where the answers to $Q$ 
are the same as those of $Q^m$ over the original database.
For each step of the contraction, we make a modification of the database.
If an atom $S(\mathbf{Y})$ was removed because it was absorbed by another atom $R(\mathbf{X})$,
then we create the relation $S$ by copying $\pi_{\mathbf{Y}} (R)$.
Note that we are allowed to create $S$ without restrictions because $Q$ has no self-joins, 
hence the database does not already contain the relation.
If a variable $v$ was removed because it was absorbed by another variable $u$,
then we extend all the relations that $u$ appears in with another attribute $v$ that takes the constant $\bot$ value everywhere and has weight $w_v(\bot) = 0$.
After projecting away the new variable, this construction does not change the query answers or their weights.

The above reductions take linear time, which is dominated by $T_S(n)$
since $T_S(n)$ is trivially in $\Omega(n)$ for the selection problem.
\end{proof}

To prove \Cref{thm:sum-selection-dicotomy} we first limit our attention to
the class of full CQs (for them $\mh(Q)=\mhfree(Q)$)
and prove the positive part in \cref{sec:sum_selection_tractable} 
and the negative part in \cref{sec:sum_selection_intractable}.
We then extend those results to more general CQs with projections in \cref{sec:sum_selection_projections}.

\subsection{Tractability Proofs for Full CQs}
\label{sec:sum_selection_tractable}

In this section, we provide tractability results for full CQs with $\mh(Q) \leq 2$.
First, we consider the trivial case of $\mh(Q) = 1$ where the maximal contraction of
$Q$ has only one atom. The lemma below is a direct consequence of the linear-time
array selection algorithm of Blum et al.~\cite{blum73select}.

\begin{lemma}
\label{lem:selection_1}
For a full CQ $Q$ with $\mh(Q) = 1$, selection by SUM is possible in 
$\Comp{1}{n}$.
\end{lemma}
\begin{proof}
By \Cref{lem:absorbed}, it suffices to solve selection on the query
$Q(x) \datarule R(x)$, 
which is a maximal contraction of all queries with $\mh(Q) = 1$, up to renaming.
Trivially, the weights of the single attribute can also be viewed as tuple weights.
Thus, applying linear-time selection \cite{blum73select} on the tuples of $R$ 
gives us the $k^\textrm{th}$ smallest query answer.
\end{proof}

For the $\mh(Q) = 2$ case, we rely on an algorithm by
Frederickson and Johnson~\cite{frederickson84selection},
which generalizes selection on the X+Y problem.
If the two sets $X$ and $Y$ are given sorted, then the pairwise sums can be represented as a sorted matrix.
A \emph{sorted matrix} $M$ contains a sequence of non-decreasing elements
in every row and every column.
For the $X+Y$ problem, a cell $M[i, j]$ contains the sum $X[i] + Y[j]$.
Even though the matrix $M$ has quadratically many cells, 
there is no need to construct it in advance given that we can compute each cell in constant time.
Selection on a union of such matrices $\{ M_1, \ldots, M_\ell \}$ asks for the $k^\textrm{th}$ smallest element among the cells of all matrices.

\begin{theorem}[\cite{frederickson84selection}]
\label{th:sorted_matrices}
Selection on a union of sorted matrices $\{ M_1, \ldots, M_\ell \}$,
where $M_m$ has dimension $p_m \times q_m$ with $p_m \geq q_m$, is possible in
time $\bigO(\sum_{m=1}^\ell q_m \log (2 p_m / q_m))$.
\end{theorem}

Leveraging this algorithm, we provide our next positive result:

\begin{lemma}\label{lem:sum-selection-binary}
For a full CQ $Q$ with $\mh(Q) = 2$, 
selection by SUM is possible in 
$\Comp{1}{n \log n}$.
\end{lemma}

\begin{proof}
The maximal contraction of full CQs with $\mh(Q) = 2$ is
$Q_1(x, z) \datarule R(x), S(z)$
or $Q_2(x, y, z) \datarule R(x, y), S(y, z)$, up to renaming.
Thus by \Cref{lem:absorbed}, it is enough to prove an $\bigO(n \log n)$ bound for these two queries.
As before, we turn the attribute weights into tuple weights.
For $Q_1$, the attribute weights are trivially tuple weights
and for $Q_2$, we assign each
attribute weight to only one relation to avoid double-counting. 
Thus, for $Q_2$ we compute 
$w(r) = w_x(r[x]) + w_y(r[y])$ and
$w(s) = w_z(s[z])$
for all $r \in R$ and $s \in S$, respectively.
Since the query is full, the weights of the query answers are in one-to-one
correspondence with the pairwise sums of weights of tuples from $R$ and $S$.

For $Q_2$, we group the $R$ and $S$ tuples by their $y$ values:
we create $\ell$ buckets of tuples where all tuples $t$ within a bucket have
equal $t[y]$ values. This can be done in linear time.
For $Q_1$, we place all tuples in a single bucket.
For each assignment of a $y$ value (no assignment for the case of $Q_1$), 
the query answers with those values
are formed by the Cartesian product of $R$ and $S$ tuples inside that bucket.
Also, if the size of bucket $m$ is $n_m$, then $n_1 + \cdots + n_\ell = |R| + |S| = \bigO(n)$.
We sort the tuples in each bucket (separately for each relation)
according to their weight in $\bigO(n \log n)$ time.
Assume $R_m$ and $S_m$ are the partitions of $R$ and $S$ in
bucket $m$ and $R_m[i]$ denotes the $i^\textrm{th}$ tuple of $R_m$ in sorted order 
(equivalently for $S_m[j]$).
We define a union of sorted matrices $\{ M_1, \ldots, M_\ell \}$ by setting
for each bucket $m$: $M_m[i, j] = w(R_m[i]) + w(S_m[j])$ if $|R_m| \geq |S_m|$
or $M_m[i, j] = w(S_m[i]) + w(R_m[j])$ otherwise 
(this distinction is needed simply to conform with the way \Cref{th:sorted_matrices} is stated).
Selection on these matrices is equivalent to selection on the query answers of $Q$.
By \Cref{th:sorted_matrices}, if matrix $M_m$ has dimension $p_m \times q_m$ with $p_m \geq q_m$, 
we can achieve selection in 
$\bigO(\sum_{m=1}^\ell q_m \log (2 p_m / q_m)) \subseteq
\bigO(\sum_{m=1}^\ell q_m \cdot 2 p_m / q_m) = 
\bigO(\sum_{m=1}^\ell p_m) = 
\bigO(\sum_{m=1}^\ell n_m) = 
\bigO(n)$.
Overall, the time spent is $\bigO(n \log n)$ because of sorting.
\end{proof}

\subsection{Intractability Proofs for Full CQs}
\label{sec:sum_selection_intractable}

Though selection is a special case of direct access,
we show that for most full CQs, time
$\bigO(n \polylog n)$ is still unattainable.
We start from the cases covered by \Cref{lem:3sum_to_rra_geq3}.
To extend that result to the selection problem, note that
a selection algorithm can be repeatedly applied for solving 
direct access. For queries with three free and independent variables,
an $\bigO(n^{2-\epsilon})$ selection algorithm
would imply a $\Comp{1}{n^{2-\epsilon}}$
direct-access algorithm, 
which we showed to be impossible.
Therefore, the following immediately follows from \Cref{lem:3sum_to_rra_geq3}:

\begin{corollary}
\label{col:selection_geq3}
If a full CQ $Q$ is self-join-free and $\freeind(Q) \geq 3$, 
then selection by SUM is not possible in %
$\Comp{1}{n^{2-\epsilon}}$
for any $\epsilon > 0$ assuming \ThreeSUM{}.
\end{corollary}

This leaves only a small fraction of full acyclic CQs to be covered: queries with 
two or fewer independent variables and three or more maximal hyperedges.
We next show that these queries 
all contain a length-3 chordless path\footnote{
The conference version of this paper~\cite{carmeli21direct} erroneously claims that the maximal contraction
of these CQs is the 3-path query $Q(x, y, z, u) \datarule R(x, y), S(y, z), T(z, u)$.
This is not correct because for example, $Q'(x, y, z, u, b) \datarule R(x, b, y), S(y, b, z), T(z, b, u)$ also satisfies $\freeind(Q') < 3$ and $\mh(Q') > 2$.
However, as we show here, the same reduction that was used for the 3-path query $Q$ can also work for any of these CQs.
},
a property that we will use in order to prove a lower bound.

\begin{lemma}\label{lem:3path_new}
The hypergraph of the maximal contraction of any full acyclic CQ with $\freeind(Q) < 3$ and $\mh(Q) > 2$
contains a chordless path of four variables.
\end{lemma}
\begin{proof}
First, for $\freeind(Q) = 1$,
we have by \Cref{lem:freeind} that an atom contains all free variables, thus $\mh(Q) = 1$. 
For the case of $\freeind(Q) = 2$, 
let $x$, $y$ be independent variables.
We distinguish two cases.

\underline{Case 1:} each of $x$ and $y$ appear in exactly one maximal hyperedge.

Denote the maximal hyperedge containing $y$ by $e_y$ and the maximal hyperedge containing $x$ by $e_x$.
Since there are at least $3$ maximal edges, there is a hyperedge $e_0$ such that $x,y\not\in e_0$.
Since $e_0$ is not absorbed by $e_y$, there exists $a\in e_0$ and $a\not\in e_y$. Thus, $a$ and $y$ are not neighbors.
Since $\{a,x,y\}$ is not an independent set, we conclude that $a$ and $x$ are neighbors,
and because only one maximal hyperedge can contain $x$, we get that $a \in e_x$.
Similarly, since $e_0$ is not absorbed by $e_x$, we conclude that there exists a neighbor $b$ of $y$ such that $b\in e_0$, $b\not\in e_x$.
Overall, we have a chordless path $x-a-b-y$.

\underline{Case 2:} $x$ or $y$ (or both) appear in at least two hyperedges.

Assume WLOG that two maximal hyperedges contain $x$.
According to the claim we shall prove next, this means that there exist non-neighbors $a$,$b$ that are both neighbors of $x$.
Then, since $\{a,b,y\}$ is not an independent set, we have that $y$ is a neighbor of $a$ or $b$. 
Assume WLOG it is $a$. 
Then, we have a path $y-a-x-b$ such that $y$ and $x$ are not neighbors and $a$ and $b$ are not neighbors. 
We conclude that also $y$ and $b$ are not neighbors, otherwise this path is a chordless cycle contradicting acyclicity.

\underline{Claim:}
If two maximal hyperedges in an acyclic hypergraph both contain a node $x$, then there exists a chordless path $a-x-b$ for some vertices $a, b$.

\underline{Proof:}
If there are two non-neighboring neighbors of $x$, then we are done.
For the sake of contradiction, assume that there are no such neighbors.
We prove with induction on the number of neighbors that for every $k$ neighbors of $x$, 
there exists a hyperedge that contains all of them and $x$.
For the base of the induction, consider $2$ neighbors of $x$. 
By our assumption, they must necessarily be neighbors. 
Since the graph is acyclic, every triangle must be covered by a hyperedge, 
so there exists a hyperedge containing both of them and $x$. 
For the inductive step, consider $k$ neighbors of $x$.
By the induction hypothesis, every subset of these neighbors of size $k-1$ appears in a hyperedge together with $x$. 
If there is no hyperedge that contains all $k$ of them, then these neighbors (without $x$) form a $(k,k-1)$-hyperclique contradicting acyclicity. 
If there is, then these $k$ neighbors along with $x$ form a $(k+1,k)$-hyperclique. 
This contradicts acyclicity unless there is a hyperedge containing all $k$ variables and $x$.
This concludes the induction.
Now consider all neighbors of $x$. 
By the induction, one edge contains all of them and $x$. 
This contradicts the fact that $x$ appears in two maximal hyperedges.

\end{proof}

Now that we established the precise form of the queries we want to classify,
we proceed to prove their intractability.
We approach this in a different way than the other hardness proofs: instead of
relying on the \ThreeSUM{} hypothesis, we instead show that tractable selection
would lead to unattainable bounds for Boolean cyclic queries.

\begin{lemma}\label{lem:sum-selection-3path}
If a full CQ $Q$ is self-join-free and the hypergraph of its maximal contraction contains a chordless path of four variables,
then selection by SUM is not possible in
\Comp{1}{n \polylog n}
assuming \hyperclique{}.
\end{lemma}
\begin{proof}
We will show that if selection for $Q$ can be done in 
$\bigO(n \polylog n)$, then the Boolean triangle query 
can be evaluated in the same time bound,
which contradicts the \hyperclique{} hypothesis.
Let $Q_\triangle() \datarule R'(x', y'), S'(y', z'), T'(z', x')$
be a query over a database $I'$ of size $\bigO(n)$.
We will construct a database $I$ for $Q$ so that
weight lookup for $Q$ over $I$ will allow us to answer $Q_\triangle$ over $I'$.

Let $x-y-z-u$ be the chordless path in the hypergraph of the maximal contraction of $Q$.
This implies that there are atoms 
$R(x, y, \mathbf{X}_R), S(y, z, \mathbf{X}_S), T(z, u, \mathbf{X}_T)$ in $Q$.
We construct a database $I$ where in all relations,
we let $x$ and $u$ respectively take all the values that $x'$ can take in $I'$.
We repeat the same for $y$ with $y'$ and $z$ with $z'$.
The values that all the other variables can take in $I$ are set to a fixed domain value $\bot$.
Because the $x-y-z-u$ path is chordless, 
we know that the pair $x-z$ never appears in a single atom,
and so is the case for $x-u$ and $y-u$.
Therefore, the size of each relation in $I$ is bounded by $|R'|$ or $|S'|$ or $|T'|$
and the size of $I$ is thus $\bigO(n)$.
Now consider a query answer $q \in Q(I)$. 
If $\pi_u(q) = \pi_x(q)$, then $\pi_{xyz}(q)$ has to satisfy all three atoms of $Q_{\triangle}$.
This is because an $(x, y)$ pair of values has to satisfy $\pi_{xy}(R)$ which contains precisely the tuples of $R'$,
and similarly for $(y, z)$ with $\pi_{yz}(S)$.
For a $(z, x)$ pair of values, these are the same as $(z, u)$ and satisfy $\pi_{zu}(T)$
contains precisely the tuples of $T'$.

We now assign weights as follows:
If $\dom \subseteq \R$, then
$w_x(i) = i, w_u(i) = -i$, and for all other variables $t$,
$w_t(i) = 0$.
Otherwise, it is also easy to assign $w_x$ and $w_u$ in a way s.t.
$w_x(i) = w_x(j)$ if and only if $i = j$ and $w_u(i) = - w_x(i)$.
This is done by maintaining a lookup table for all the domain values that we map to some arbitrary real number.
Then, we perform weight lookup for $Q$ to identify if a query result with zero weight
exists.
If it does for some result $q$, then 
$w_x(\pi_x(q)) + \ldots + w_u(\pi_u(u)) = 0$ hence
$\pi_x(q) = \pi_u(q)$ and $Q_\triangle$ is true, otherwise it is false.
If the time to access the sorted array of $Q$-answers takes $\bigO(n \polylog n)$,
then by \Cref{lem:inverted_sum} weight lookup also takes $\bigO(n \polylog n)$,
contradicting \hyperclique{}.
\end{proof}

For full CQs, the negative part of \Cref{thm:sum-selection-dicotomy} for acyclic queries is proved by combining \Cref{col:selection_geq3} and \Cref{lem:sum-selection-3path} together with \Cref{lem:3path_new} and \Cref{lem:absorbed} which show that we cover all queries.
For self-join-free cyclic CQs, 
we once again resort to the hardness of their Boolean version based on \hyperclique{}.

We summarize below what we have proven so far for full CQs:

\begin{lemma}\label{lem:sum-selection-full-dicotomy}
Let $Q$ be a full CQ.
\begin{itemize}
    \item If $\mh(Q) \leq 2$,
    then selection by SUM is possible in 
    $\Comp{1}{n \log n}$.
    \item Otherwise, if $Q$ is also self-join-free, 
    then selection by SUM is not possible in %
    $\Comp{1}{n \polylog n}$.
    assuming \ThreeSUM{} and \hyperclique{}.
\end{itemize}
\end{lemma}

\subsection{CQs with Projections}
\label{sec:sum_selection_projections}

To complete the proof of \cref{thm:sum-selection-dicotomy}, 
we now show how the results from \cref{sec:sum_selection_tractable,sec:sum_selection_intractable} generalize to free-connex CQs.
For that purpose, we mainly rely on \cref{prop:reduce-to-subtree}, 
which allows us to reduce a free-connex CQ to a full CQ.
One difficulty that we encounter is that the full CQ that we obtain 
by this reduction may not
necessarily be unique. Thus, the following definition will be useful:

\begin{definition}[Reduced Full CQs]
\label{def:equivalent_full_queries}
For a free-connex CQ $Q$, $\RF(Q)$ is the set of all possible full CQs
we can obtain by the reduction of \cref{prop:reduce-to-subtree}.
\end{definition}

We now show that for our purposes, the tractability of a free-connex CQ $Q$ coincides with
that of any of the CQs in $\RF(Q)$.

\begin{lemma}\label{lem:selection_free_connex}
Let $Q$ be a free-connex CQ and $Q' \in \RF(Q)$.
Selection for $Q$ by SUM is possible in $\Comp{1}{T_S(n)}$ if
selection for $Q'$ by SUM is possible in $\Comp{1}{T_S(n)}$.
The converse is also true if $Q$ is self-join-free.
\end{lemma}
\begin{proof}
The ``if'' direction is trivial from \cref{prop:reduce-to-subtree}:
If $Q$ is over a database $I$, we can use $Q'$ over a modified database $I'$ to obtain exactly the same answers.

For the ``only if'' direction, 
we use selection on $Q$ to answer selection on $Q'$. 
If $Q'$ is over a database $I'$, then we construct a modified database $I$ for $Q$ as follows.
We copy all the relations of $I'$ into $I$ and for every existential variable of $Q$, we
add an attribute to the corresponding relations that takes the same $\bot$ value in all tuples.
The weight of all the new attributes is set to $0$ for $\bot$.
Now, the answers to $Q$ over $I$ are the same as those of $Q'$ over $I'$ if we ignore all the $\bot$ values from the answers.

The above reductions take linear time, which is dominated by $T_S(n)$
since $T_S(n)$ is trivially in $\Omega(n)$ for the selection problem.
\end{proof}

From \Cref{lem:sum-selection-full-dicotomy}, we can decide the tractability of any (self-join-free) full CQ 
in $\RF(Q)$ by the number of its maximal hyperedges.
We now connect this measure to the free-maximal-hyperedges $\mhfree(Q)$,
which is a measure easily computable from the original query $Q$.

\begin{lemma}\label{lem:maximal_free_hyperedges}
For a free-connex CQ $Q$, $\mhfree(Q) = \mh(Q')$ for any $Q' \in \RF(Q)$.
\end{lemma}
\begin{proof}
Let $T$ be the ext-$\free(Q)$-connex tree used to derive $Q'$ and $T'$ be the connected subtree containing exactly the free variables.
Recall that the nodes of $T'$ correspond to the atoms of $Q'$.

First, we prove that every node of $T'$ is a subset of some hyperedge of $\calH_\free(Q)$ and as a result $\mhfree(Q) \leq \mh(Q')$.
Consider a node of $T'$ and let $V$ be the corresponding set of variables.
Note that $V$ are all free variables since they appear in $T'$.
Since $T'$ is a subtree of $T$, which in turn is a join-tree of an inclusive extension of $\calH(Q)$, 
there must exist an atom in $Q$ that contains $V$.
Let the set of variables of that atom be $V \cup V_1 \cup X$, 
for some disjoint sets $V, V_1, X$,
where $V_1$ are free variables and $X$ are existential.
By the definition of $\calH_\free(Q)$, 
there must exist a hyperedge $V \cup V_1$ in $\calH_\free(Q)$.

Second, we prove that every maximal hyperedge in $\calH_\free(Q)$ is a subset of some node of $T'$ and as a result $\mh(Q') \leq \mhfree(Q)$.
Let $V$ be the (free) variables of a hyperedge of $\calH_\free(Q)$. 
Then, there must exist an atom $e$ in $Q$ that contains the $V$ variables.
If $e$ corresponds to a node in $T'$, then we are done. 
Otherwise, $e$ corresponds to a node in $T \setminus T'$. 
Let $V'$ be the first node of $T'$ on the path from $V$ to the root.
A path between $V$ and any node of $T'$ must necessarily pass through $V'$,
otherwise $T$ would contain a cycle.
Since $V$ contains only free variables, 
each one of them has to appear in some node of $T'$
and by the running intersection property, in all the nodes on the path to $V$.
Therefore, $V'$ contains all the variables of $V$.
\end{proof}

\Cref{lem:maximal_free_hyperedges,lem:selection_free_connex} allow us to generalize the positive part of \Cref{lem:sum-selection-full-dicotomy} 
from full CQs to free-connex CQs by simply replacing the $\mh(Q)$ measure with $\mhfree(Q)$.
For the negative part, we need three ingredients to cover the class of self-join-free CQs: 
The case of free-connex CQs with $\mhfree(Q) > 2$ is also covered by the lemmas above.
For acyclic CQs that are not free-connex, we use \Cref{lem:lex-selection-intractability} which applies to any ranking function, including SUM.
Finally, the intractability of cyclic CQs follows from \hyperclique{}.

\def\qext{Q^+}
\def\fds{\Delta}
\def\fdsext{\Delta^+}

\section{Functional Dependencies}\label{sec:FDs}

In this section, we extend our results to apply to databases that are constrained by Functional Dependencies (FDs).
From the point of view of a CQ, if the allowed input databases are restricted to satisfy an FD, an assignment to some of the variables uniquely determines the assignment to another variable.
Our positive results are not affected by this, since an algorithm can simply ignore the FDs.
However, certain CQs that were previously intractable may now become tractable in the presence of FDs.
Our goal is to investigate how the tractability landscape changes for
all four variants we have investigated so far:
direct access and selection by LEX or SUM orders.

\introparagraph{Concepts and Notation for FDs}
In this section, we assume that the database schema $\calS$ is extended with FDs
of the form $R: A \rightarrow B$, where $A$ and $B$ are sets of integers.
This means that if the tuples of relation $R$ agree on the attributes indexed by $A$,
then they also agree on those indexed by $B$.
Though FDs are usually defined directly on the schema of the database, we express them from now on using the query variables for convenience.
More specifically, we assume that an FD has the form $R: \mathbf{X} \rightarrow \mathbf{Y}$ for some 
atom $R(\mathbf{Z})$
where $\mathbf{X}, \mathbf{Y} \subseteq \mathbf{Z}$. 
We now briefly explain why this assumption can be done without loss of generality.
There are two factors that can render such a notation not well-defined. The first is self-joins. However, our negative results apply only to self-join-free CQs regardless of this issue, and positive results for self-join-free CQs naturally extend to CQs with self joins: 
Self-joins can be reduced in linear time to a self-join-free form
by replacing the $i^\textrm{th}$ occurrence of a relational symbol $R$ with a fresh relational symbol $R_i$
and then copying relation $R$ into $R_i$. 
The second factor is repeated appearances of a variable in an atom (e.g., $R(x, x)$). Such an appearance can be eliminated in a linear-time preprocessing step that performs the selection on the relation and then removes the duplicate variable. For the other direction of the equivalence, the projected relation can be transformed to one with the repeated variable by duplicating the relevant column in the 
relation.
We say that an FD 
$\mathbf{X} \rightarrow \mathbf{Y}$
is \emph{satisfied} by an input database $I$ if for all tuples $t_1, t_2 \in R^I$
we have that if $t_1[x] = t_2[x], \forall x \in \mathbf{X}$, then 
$t_1[y] = t_2[y], \forall y \in \mathbf{Y}$.
For such an FD, we sometimes say that $\mathbf{X}$ \emph{implies} $\mathbf{Y}$.
WLOG, we assume that all FDs are of the form $R: \mathbf{X} \rightarrow y$ 
where $y$ is a single variable
because we can replace an FD of the form $R: \mathbf{X} \rightarrow \mathbf{Y}$
with a set of FDs $\{ R: \mathbf{X} \rightarrow y \;|\; y \in \mathbf{Y} \}$.
If $|\mathbf{X}| = 1$, we say that the FD is \emph{unary}.
In this paper, we only consider unary FDs.

\introparagraph{Complexity and Reductions}
In all previous sections, the computational problem we were considering was 
defined by a CQ and we were interested in the worst-case complexity over all possible input databases.
Now, the problem is defined by a CQ $Q$ and a set of FDs $\Delta$.
Thus the input is limited to databases that satisfy all the FDs $\Delta$.
Therefore, to determine the hardness of a problem, 
we need to consider the complexity of the combination of CQ and FDs.
In particular, the results in the previous sections apply when the given FD set is empty. 
Our reductions from now on are from a certain CQ and FD set to another CQ and another FD set.

\begin{definition}[Exact Reduction]
We say that there is an exact reduction from a CQ $Q$ with FDs $\Delta$ to a CQ $Q'$ with FDs $\Delta'$ 
if for every database $I$ that satisfies $\Delta$:
\begin{enumerate}
    \item We can construct a database $I'$ in $\bigO(|I|)$ that satisfies the FDs $\Delta'$.
    \item There is a bijection $\tau$ from $Q'(I')$ to $Q(I)$ that is computable in $\bigO(1)$.
\end{enumerate}
Additionally, we say that an exact reduction is \emph{weight-preserving} if for every weight function $w$ given for $Q$,
there is a weight function $w'$ for $Q'$ such that $w(\tau(q')) = w'(q'), \forall q' \in Q'(I')$.
\end{definition}

\introparagraph{Known Results}
Carmeli and Kr{\"{o}}ll~\cite{DBLP:journals/mst/CarmeliK20} reasoned about the complexity of
enumerating the answers to a CQ 
by looking at an equivalent extended CQ over an extended schema and FDs.
We recall the definition of this extension here\footnote{The original definition is more involved. We only give the restriction of the definition required for our purposes.} and then proceed to use it for our classification.

\begin{definition}[FD-extension~\cite{DBLP:journals/mst/CarmeliK20}]\label{def:fd-extension}
Given a self-join free CQ $Q$ and a set of FDs $\Delta$,
we define two types of extension steps.
For an FD $R: \mathbf{X} \rightarrow y$:
\begin{enumerate}
    \item If $\mathbf{X} \subseteq \mathbf{Z}$ for some atom $S(\mathbf{Z})$ and
    $y \notin \mathbf{Z}$,
    then increase the arity of $S$ by one, replace $S(\mathbf{Z})$ with $S(\mathbf{Z}, y)$,
    and add $S: \mathbf{X} \rightarrow y$ to the FD set.
    \item If $\mathbf{X} \subseteq \free(Q)$ and $y \notin \free(Q)$, then add $y$ to $\free(Q)$.
\end{enumerate}
The FD-extension of $Q$ and $\Delta$ is a CQ $Q^+$ and a set of FDs $\Delta^+$
that are obtained as the fixpoint of the above two extension steps.

\end{definition}

\begin{example}
Consider the CQ $Q_{2P}(x, z) \datarule R(x, y), S(y, z)$.
This CQ is not free-connex, therefore selection (or direct access) is not possible
by any order (\Cref{lem:lex-selection-intractability}).
Now, if we know that the database satisfies the FD $S: y \rightarrow z$,
we can take the unique $z$-value for every $y$ from $S$ and add it to every tuple of $R$, 
while preserving the same query answers.
Thus, we can extend the CQ to $Q_{2P}^+(x, z) \datarule R(x, y, z), S(y, z)$
and add the FD $R: y \rightarrow z$.
While the original CQ was not free-connex, notice that $Q_{2P}^+$ now is.
This makes it tractable for all the tasks that we consider in this article
since it is acyclic and $R$ contains all the free variables (see \Cref{sec:sum_direct}).

Similarly, the FD-extension may transform a cyclic CQ into an acyclic one.
For $Q_\triangle(x, y, z) \datarule R(x, y), S(y, z), T(z, x)$ with the FD $S: y \rightarrow z$,
we get the FD-extension $Q_\triangle^+(x, y, z) \datarule R(x, y, z), S(y, z), T(z, x)$ with the additional FD $R: y \rightarrow z$
which is acyclic because $R$ contains all the variables.
Like $Q_{2P}^+(x, z)$, we have that $Q_\triangle^+$ is tractable for all the tasks that we consider. 
This is despite the fact that the cyclic $Q_\triangle$ without FDs is intractable for all of these tasks.
\end{example}

Previous work~\cite{DBLP:journals/mst/CarmeliK20} showed exact reductions between the original CQ (with the original FDs) and its extension (with the extended FDs) in both directions, proving that the two tasks are essentially equivalent for the task of enumeration.

\begin{theorem}[\cite{DBLP:journals/mst/CarmeliK20}]\label{thm:fds-known-old}
Let $Q$ be a self-join free CQ 
with FDs $\Delta$.
There are exact reductions between $Q$ with $\Delta$ and $Q^+$ with $\Delta^+$ in both directions.
\end{theorem}

This theorem alone implies that, if $Q^+$ has a tractable structure (free-connex for enumeration), 
then the original CQ $Q$ is tractable too. 
That is because according to the definition of an exact reduction,
an instance for the extended schema can be built in linear time, and answers to
the extension can be translated back to answers of the original CQ in constant time per answer.

We restate this result below in a slightly more general way,
adding the fact that
these reductions are actually weight-preserving. 
Intuitively, this means that any SUM ordering of the answers in one problem can be preserved
through the reduction.

\begin{lemma}\label{thm:fds-known}
Let $Q$ be a self-join free CQ 
with FDs $\Delta$.
There are weight-preserving exact reductions between $Q$ with $\Delta$ and $Q^+$ with $\Delta^+$ in both directions.
\end{lemma}
\begin{proof}
The bijections between the answers in the exact reductions do not change the values of the free variables of $Q$. 
For the reduction from the query to its extension, we set the weights of the new free variables to zero. 
For the reduction from the extension to the original query, if the extension has a free variable $y$ that is existential in the original query, then some free variable $x$ in the original query implies $y$. 
To maintain the contribution of the $y$-weight in the query answers,
we simply increase the weight of every domain value of $x$ by adding the weight of the corresponding domain value of $y$ 
(i.e., the one that is implied by the FD).
\end{proof}

Proving lower bounds using the extension is more tricky.
If the extension has a structure that is known to make a CQ intractable 
(e.g. not free-connex for enumeration), 
it is not necessarily the case that it is intractable together with the FDs.
Still, Carmeli and Kr{\"{o}}ll~\cite{DBLP:journals/mst/CarmeliK20} were able to prove lower bounds 
for the cases where the extension is not free-connex. 
We proceed to use the same technique that they used in order to prove a more general statement.

\introparagraph{Eliminating FDs}
Our goal is to show that for the extension of a self-join-free CQ, 
unary FDs essentially do not affect our classification.
This result is similar in spirit to other results in database theory where the complexities of problems with FDs were also identified to be the complexity of the original complexity criterion applied to the FD extension \emph{after removing the FDs}
\cite{DBLP:conf/pods/Kimelfeld12,
DBLP:journals/vldb/GatterbauerS17,
DBLP:journals/pvldb/FreireGIM15}.
To achieve that, we reduce a CQ without FDs to the same CQ with FDs,
under the condition that the CQ and the FDs we reduce to are the
extension of some CQ.

\begin{lemma}\label{lemma:extension-to-no-FDs}
Let $Q$ be a self-join free CQ 
with unary FDs $\Delta$.
There is a weight-preserving exact reduction from $Q^+$ 
without FDs to $Q^+$ over $\Delta^+$.
\end{lemma}

\begin{proof}
We are given a database instance $I$ for $Q^+$ that does not necessarily satisfy any FDs
and we construct another database instance $I'$ with roughly the same answers (there is a bijection $\tau$ such that $Q^+(I) = \tau(Q^+(I'))$) 
that satisfies the extended FDs $\Delta^+$.
Given a variable $v$, denote by $I_v$ the set of all variables that are transitively implied by $v$. 
By the definition of the extension, for every variable $v$, every atom of $Q^+$ that contains $v$ also contains all $I_v$ variables.
For every tuple in the relation that corresponds to such an atom, 
we replace the $v$-value with the concatenation of $I_v$-values. 
We set the weight of a concatenated value to be equal to the weight of the $v$ variable as it was in the original database $I$.
This construction can be done in linear time, and the resulting database satisfies the FDs.

We now claim that this construction preserves the answers through a bijection.
Note that for every free variable $v$, by the definition of the extension, the variables in $I_v$ are all free.
Every answer to the original problem gives an answer to our construction by assigning every free variable $v$ to the concatenation of the assignments of $I_v$.
Every answer to our construction gives an answer to the original problem by keeping only the value that corresponds to the original variable for every free variable.
In both cases, the weights of the query answers are the same in the two instances.
\end{proof}

\begin{example}
We illustrate the reduction above through the CQ 
$Q(x, z, u) \datarule R(x, y),$ $S(y, z), T(z, u)$ 
with FD $T: z \rightarrow u$.
The FD-extension is $Q^+(x, z, u) \datarule R(x, y), S'(y, z, u), T(z, u)$
with extended FDs $T: z \rightarrow u$ and $S': z \rightarrow u$.
Notice that $Q^+$, like $Q$, is not free-connex. 
From that structural property of $Q^+$, we know that without any FDs it is
intractable,
e.g., for the task of selection by \Cref{lem:selection_free_connex}.
The reduction shows that this is still the case, even if we take the extended FDs into account.

Concretely, the reduction takes a database that does not satisfy the extended 
FDs and constructs another database that does.
For example, $S'$ could contain tuples $(1, 1, 1)$ and $(1, 1, 2)$ where the middle attribute (corresponding to $z$) 
does not imply the third attribute (corresponding to $u$).
To modify this database, we replace the $z$-values by values that pack $z$ and $u$ together
as $(z,u)$.
Conceptually, and abusing the notation a little, 
$S'(y,z,u)$ is further replaced with $S''(y,(z, u), u)$
and the whole query with
${Q^+}'(x, (z , u), u) \datarule R(x, y), S''(y, (z ,u) , u), T'((z , u), u)$.
The aforementioned two tuples of $S'$ are thus replaced by $(1, (1, 1), 1)$ and $(1, (1, 2), 2)$.
Now, the third attribute is trivially dependent on the middle one because it is contained in the latter.
This reduction is only possible because, being an FD-extension of $Q$,
$Q^+$ always contains $u$ in all of its atoms that contain $z$.
The weights of the query answers are preserved if we only keep the weight of $z$ in the concatenated $(z, u)$ values, i.e., $w_{z}((1,1)) = w_{z}(1)$ and $w_{z}((1,2)) = w_{z}(1)$.
\end{example}

Note that a reduction in the opposite direction is trivial since, given an instance that satisfies the FDs, 
it is also a valid instance for the instance without the FDs. 
Thus, this lemma proves that the two problems are equivalent.
By combining this fact with the equivalence of CQs and their FD-extensions
(\cref{thm:fds-known}),
we obtain the following result that is also useful for lower bounds 
and allows us to classify queries based on the structure of their extension.

\begin{theorem}\label{thm:equivalence-fds-extension}
Let $Q$ be a self-join free CQ 
with unary FDs $\Delta$.
There are weight-preserving exact reductions between $Q$ 
over $\Delta$ 
and $Q^+$ 
without FDs in both directions.
\end{theorem}

\subsection{Sum of Weights}
\label{sec:fds_sum}

For a sum-of-weights order, applying a weight-preserving exact reduction
cannot reorder the query answers since their weight is preserved.
More formally, the classes of self-join-free CQs that are tractable for selection
and direct access
are both \emph{closed} under weight-preserving exact reductions.
Therefore, \Cref{thm:equivalence-fds-extension} immediately proves that
for both problems, we can classify self-join-free CQs by their FD-extension.

\begin{theorem}
\label{th:fd_sum_dichotomy}
Let $Q$ be a CQ with unary FDs $\Delta$.
\begin{itemize}
    \item If $Q^+$ is acyclic and an atom of $Q^+$ contains all the free variables,
    then direct access for $Q$ by SUM is possible 
    in $\Comp{n \log n}{1}$.
    \item Otherwise, if $Q$ is also self-join-free, 
     direct access for $Q$ by SUM is not possible 
     in $\Comp{n \polylog n}{\polylog n}$,
     assuming \ThreeSUM{} and \hyperclique{}.
 \end{itemize}
 \end{theorem}
 \begin{theorem}
Let $Q$ be a CQ with unary FDs $\Delta$.
\begin{itemize}
    \item If $Q^+$ is free-connex and $\mhfree(Q^+) \leq 2$,
    then selection for $Q$ by SUM is possible in 
    $\Comp{1}{n \log n}$.
    \item Otherwise, if $Q$ is also self-join-free, 
    then selection for $Q$ by SUM is not possible in %
    $\Comp{1}{n \polylog n}$,
    assuming \ThreeSUM{}, \hyperclique{}, and \seth{}.
\end{itemize}
\end{theorem}

\subsection{Lexicographic Orders}

We now move on to lexicographic orders, where we also provide dichotomies for unary FDs.
The analysis gets more intriguing compared to SUM since 
the FDs may interact with the given lexicographic order in non-trivial ways
that have not been explored by any previous work we know of.
Like before, we use as our main tool an extension of the query according to the FDs and show an equivalence between the extension and the original problem
with respect to tractability.
The key difference now is that the extension may also \emph{reorder} the variables in the lexicographic order according to the FDs.

\subsubsection{Tractability Results}

Our positive results rely on two ingredients.
The first ingredient is to have a reduction that preserves precisely the same lexicographic order that we begin with.
This is necessary because if the query answers after the reduction follow a different lexicographic order,
then we cannot conclude anything about the position of a query answer in the original problem.
Thus, we define a stricter notion of an exact reduction.

\begin{definition}[Lex-preserving exact reduction]
Consider a CQ $Q$ and FDs $\Delta$, and a CQ $Q'$ and FDs $\Delta'$ such that the free variables of $Q$ are also free variables in $Q'$.
An exact reduction via a bijection $\tau$ from $Q$ and $\Delta$ to $Q'$ and $\Delta'$ is called \emph{lex-preserving} if, for every partial lexicographic order $\lex$ for $Q$
and for all query answers $q_1, q_2$ of $Q'$,
we have that $q_1 \prec' q_2$ iff $\tau(q_1) \preceq \tau(q_2)$ 
where $\preceq$ and $\preceq'$ are orders implied by $\lex$.
\end{definition}

In a lex-preserving exact reduction, we have the guarantee that even though
the query answers might not be exactly the same,
those that are in a 1-1 correspondence will be placed in the same position when ordered.
Thus, such a reduction allows us to solve direct access or selection on a different problem
and translate the answers back to the original problem.
Analogously to weight-preserving exact reductions for SUM,
our notions of tractability for lexicographic orders
are preserved under lex-preserving exact reductions.

Now notice that the forward reduction from $Q$ with $\Delta$ to $Q^+$ with $\Delta^+$
in \Cref{thm:fds-known} is a lex-preserving exact reduction.
This is because as we argued before,
the reduction to the extension does not change the values of any of the free variables.
The following lemma is,
similarly to \Cref{thm:fds-known},
a slight generalization of the result of
Carmeli and Kr{\"{o}}ll~\cite{DBLP:journals/mst/CarmeliK20},
this time suited to lexicographic orders.

\begin{lemma}\label{thm:fds-known-lex}
Let $Q$ be a self-join free CQ 
with FDs $\Delta$.
There is a lex-preserving exact reduction from $Q$ with $\Delta$ to $Q^+$ with $\Delta^+$.
\end{lemma}

Importantly, the other direction of \Cref{thm:fds-known-lex} is not true in general,
since the FD-extension may contain additional free variables that do not appear
in the original query $Q$.
Similarly, the reduction in \Cref{lemma:extension-to-no-FDs} is also not lex-preserving. 
As an example, consider the simple example of
$Q(v_1, v_2, v_3) \datarule R(v_1, v_2, v_3)$ 
with the FD $R: v_1 \rightarrow v_3$. 
Sorting the constructed instance by $\angs{v_1, v_2, v_3}$ will actually result in an ordering according to $\angs{(v_1,v_3), v_2, v_3}$
which is the same as $\angs{v_1, v_3, v_2}$.
That is precisely the reason why proving lower bounds for the case of lexicographic orders is not as straightforward as the SUM case.

Going back to the positive side, we can now use \Cref{thm:fds-known-lex} to reduce our problem to its extension.
But in order to cover all tractable cases, we need to modify the FD-extension so that the FDs are also taken into account in the lexicographic order.
In particular, if a variable $u$ is implied by a variable $v$ and $u$ comes after $v$ in the order, 
then $u$ is placed right after $v$ in the reordering.

\begin{definition}[FD-reordered extension]
Given a self-join-free CQ $Q$, 
a set of unary FDs $\fds$ and a partial lexicographic order $\lex$,
their FD-reordered extension is a CQ $\qext$, a set of unary FDs $\Delta^+$
and a partial lexicographic order $\lex^+$.
We take $\qext$ and $\Delta^+$ to be as defined in \Cref{def:fd-extension},
and
$\lex^+$ is obtained from $\lex$ by applying the following reordering step iteratively from index $i=0$ until we reach the end of the list $\lex$:
Find all variables $I(L[i])$ that are transitively implied by $\lex[i]$
and place them in consecutive indexes starting from $i+1$.
Note that $\lex$ may grow in this process to contain variables that are free in $Q^+$ though they are not free in $Q$.

\end{definition}

\begin{example}
Consider the CQ $Q(v_1, v_2, v_3, v_4) \datarule R(v_1, v_3), S(v_3, v_2), T(v_2, v_4)$ with 
the FD $R: v_1 \rightarrow v_3$ and
the lexicographic order $\lex = \angs{v_1, v_2, v_3, v_4}$.
This order contains the disruptive trio $v_1, v_2, v_3$ and is intractable for direct access according to the results of \Cref{sec:lexic}.
When applying the FD-extension, we get that $Q^+ = Q$ because $v_1$ already appears with $v_3$ in $R$.
For the FD-reordered extension, we reorder the lexicographic order into $\lex = \angs{v_1, v_3, v_2, v_4}$,
which contains no disruptive trio and is tractable for direct access.
\end{example}

We next prove two important properties for the FD-reordered extension:

\begin{lemma}\label{lem:consecutive-order}
For every variable $v$ of an order $\lex^+$ of an FD-reordered extension,
all variables transitively implied by $v$ that appear after $v$ in $\lex^+$
have to appear consecutively after $v$ in $\lex^+$.
\end{lemma}
\begin{proof}
The property holds at the end of the process due to the transitivity of implication.
Consider a variable $y$ implied by $v$ that appears after $v$ in $\lex^+$ 
and assume there is a variable $x$ that is in-between $v$ and $y$ but is not implied by $v$.
We can further assume that $x$ is the first variable in the order with that property.
Now, $x$ must have been inserted at that position when handling another variable $z$ that is in-between $v$ and $x$ in $\lex^+$ and that $z \rightarrow x$.
Since $x$ is the first one not implied by $v$, we have that $v \rightarrow z$
which gives us $v \rightarrow x$, contradicting our assumption.

\end{proof}

We next show that the reordering of the variables gives the same result
because of the extended FDs $\Delta^+$.
As a consequence, we can study the complexity of the query with the FD-reordering.

\begin{lemma}\label{lemma:reordering}
Given a self-join-free CQ $Q$, 
a set of unary FDs $\fds$ and a partial lexicographic order $\lex$,
for every database $I$ that satisfies $\Delta^+$,
ordering $Q^+(I)$ by $\lex$ is the same as ordering it by $\lex^+$.

\end{lemma}
\begin{proof}
Once the value for a variable $v$ is set, a variable implied by $v$ can have at most one possible value, 
so as long as it comes after $v$ in the order,
its exact position or whether it is free cannot influence the answer ordering.
Therefore, the reordered extension is equivalent to the original extension.
\end{proof}

We now have what we need in order to show that if the reordered extension has a tractable structure,
then we can conclude tractability for the original problem.
Specifically, if the CQ $Q^+$ of the FD-reordered extension is free-connex, $L^+$-connex and has no disruptive trio, 
then we know it is tractable with respect to direct access by \Cref{thm:partial-dichotomy}. 
\Cref{lemma:reordering} shows that direct access to the extension by $L^+$
is the same as direct access by $L$, which is the order that we want. 
Finally, we use the fact that the reduction given in \Cref{thm:fds-known} preserves lexicographic orders to conclude that, 
given an input to $Q$, we can construct an FD-reordered extension
where tractable direct access for $Q^+$ by $L^+$ 
will give us tractable direct access for $Q$ by $L$.
Following the same process for selection and using \Cref{thm:lex-selection-dicotomy},
we conclude that if the CQ $Q^+$ is free-connex, then
selection for $Q$ by $L$ is tractable.

\subsubsection{Intractability Results}

We begin by some negative results that can easily be inferred from the results we have already proved in this paper or by
past work.

\begin{lemma}
Let $Q$ be a self-join-free CQ with unary FDs $\Delta$.
If $Q^+$ is not free-connex then for any ranking function,
direct access for $Q$ is not possible in $\Comp{n \polylog n}{\polylog n}$ assuming \sparseBMM{} and \hyperclique{},
and selection for $Q$ is not possible in $\Comp{1}{n \polylog n}$ assuming \seth{} and \hyperclique{}.
\end{lemma}
\begin{proof}
For direct access, the impossibility is implied by the hardness of enumeration.
If we could have direct access for $Q$, then we can also have enumeration with the same time bounds.
Then, by the exact reduction of \Cref{thm:equivalence-fds-extension},
we would be able to enumerate the answers to the non-free-connex CQ $Q^+$ with no FDs with
quasilinear preprocessing and polylogarithmic delay,
which is known to contradict \sparseBMM{} or \hyperclique{}~\cite{bagan2008computing,bb:thesis}.

For selection, we use \Cref{lem:lex-selection-to-counting}
together with the simple observation that an exact reduction preserves the number of query answers.
Indeed, by \Cref{thm:equivalence-fds-extension} we have that counting the answers to $Q$ with $\Delta$
can give us the count of query answers of $Q^+$ with no FDs.
Since a logarithmic number of selections can be used to find the latter,
the count of answers of the non-free-connex CQ $Q^+$ can be found in quasilinear time,
contradicting \seth{} (if $Q^+$ is acyclic) or \hyperclique{} (if $Q^+$ is cyclic).
\end{proof}

It is left to handle the cases that the reordering of the extension is free-connex acyclic but has a disruptive trio or is not $L^+$-connex.
The case that the extension is not $L^+$-connex can be shown using a reduction from enumeration for a non-free-connex CQ (\Cref{lemma:enum-prefix-using-access}) as we did in \Cref{sec:partial-intractable}.

\begin{lemma}
Let $Q$ be a self-join-free CQ with unary FDs $\Delta$, and let $\lex$ be a partial lexicographic order.
If $Q^+$ is acyclic but not $L^+$-connex, then direct access for $Q$ by $\lex$ is not possible in $\Comp{n \polylog n}{\polylog n}$ assuming \sparseBMM{}.
\end{lemma}
\begin{proof}
Using \Cref{lemma:enum-prefix-using-access}, it is enough to find a prefix $L'$ of $L^+$ that is closed under implication such that the extension $Q^+$ is not $L'$-connex. We simply take $L'=L^+$. Then, we conclude that there is efficient enumeration for the acyclic but not free-connex $Q^+$ with $L'$ as free variables. Since $L'$ is closed under implication, $Q^+$ is an extension of itself, so we can use \Cref{lemma:extension-to-no-FDs} to conclude that there exists efficient enumeration to $Q^+$ with $L'$ as free variables even without FDs. This is a contradiction as it is not free-connex.
\end{proof}

The case where the reordered extension has a disruptive trio $v_1,v_2,v_3$ is slightly more intricate
as we cannot directly use \Cref{lemma:extension-to-no-FDs}.
One might hope that, as we did in the proof of \Cref{lemma:hardness}, 
we would be able to take a prefix of $L^+$ that ends in $v_2$
and then apply \Cref{lemma:extension-to-no-FDs}.
However, that is not always possible here because we have the additional restriction that the prefix we pick has to be closed under implication.
This is required so that the CQ we obtain when we restrict the free variables to the prefix is a valid FD-extension. 
Unfortunately, a prefix that includes $v_2$ but not $v_3$ and is closed under implication does not necessarily exist.
That is the case when some variable implies both $v_2$ and $v_3$.

\begin{example}
Consider the CQ $Q(v_1, v_2) \datarule R(v_1, v_3), S(v_3, v_2)$ with 
the FD $S: v_2 \rightarrow v_3$ and
the lexicographic order $\lex = \angs{v_1, v_2}$.
The extended reordering is $\lex^+ = \angs{v_1, v_2, v_3}$ which contains the disruptive trio $v_1, v_2, v_3$.
To reuse our previous approach, we would want to claim that using lexicographic direct access to $Q^+$, we can enumerate the CQ with only $v_1,v_2$ as free variables (which happens to be in this case the same as $Q$ that we started with).
However, this is not a contradiction because $Q$ is not known to be hard for enumeration as it has FDs and it is not an extension ($v_2$ implies $v_3$ while $v_2$ is free and $v_3$ is existential). In fact, we cannot find any prefix $L'$ of $\lex^+$ that is closed under implication such that $Q^+$ is not $L'$-connex.
To circumvent this issue, we encode the enumeration of $Q$ without FDs into the extension by combining the binary search approach from  \Cref{lemma:enum-prefix-using-access} with the concatenation reduction from \Cref{lemma:extension-to-no-FDs}.
The difference is that before we used binary search to enumerate a prefix of the free variables, and now this prefix might stop in the middle of a variable with concatenated values. Thus, $v_2$ will be assigned the values $(v_2,v_3)$, and we will use binary search to skip over the $v_3$ values.
\end{example}

\begin{lemma}
Let $Q$ be a self-join-free CQ with unary FDs $\Delta$, and let $\lex$ be a partial lexicographic order.
If $Q^+$ is acyclic and $L^+$ contains a disruptive trio in $Q^+$, then direct access for $Q$ by $\lex$ is not possible in $\Comp{n \polylog n}{\polylog n}$ assuming \sparseBMM{}.
\end{lemma}
\begin{proof}
Consider a disruptive trio $v_1,v_2,v_3$ in $Q^+$ with respect to $L^+$.
Let $L'$ be the prefix of $L^+$ ending in $v_2$, and let $Q'$ be the query with the body of $Q^+$ and the free variables $L'$.
As $v_1,v_3,v_2$ is an $L'$-path, we know that $Q'$ is acyclic but not free-connex, and so it cannot be enumerated efficiently without FDs assuming \sparseBMM{}.

We first claim that the reordering $L^+$ is stable with respect to the first occurrences of implying variables. More precisely, let $a$ and $b$ be variables in $L^+$ such that $a$ appears before $b$ in $L^+$.
We claim that the first variable implying $b$ does not appear before the first variable implying $a$.
Indeed, consider the first variable $v_b$ implying $b$. If $v_b$ appears before $a$, by \Cref{lem:consecutive-order}, $v_b$ also implies $a$. So, the first variable that implies $a$ is $v_b$ or a variable before it.
As a consequence, due to the reordering, the first variable implying a value that appears after $v_2$ appears after all first variables implying values before (and including) $v_2$.

We can now claim that we can enumerate the answers to $Q'$ without FDs using lexicographic direct access to $Q^+$ with FDs.
We use the same construction as we did in \Cref{lemma:extension-to-no-FDs} by assigning each variable a concatenation of the variables it implies, except that now we need to be careful about the order in which we concatenate: we start with any variables in $L^+$, ordered by $L^+$. The constructed database satisfies the FDs.
It is only left to use binary search, similarly to \Cref{lemma:enum-prefix-using-access}, in order to enumerate the distinct values of the variables of $L'$.
Due to the previous paragraph, we know that these appear as a prefix, before the first value of a variable after $v_2$.
\end{proof}

The results of this section regarding lexicographic orders are summarized as follows.

\begin{theorem}\label{th:fds_lex_direct}
Let $Q$ be a CQ with unary FDs $\Delta$ and $\lex$ be a partial lexicographic order.
\begin{itemize}
    \item If $Q^+$ is free-connex and $\lex^+$-connex and does not have a disruptive trio with respect to $\lex^+$, 
    then direct access for $Q$ by $\lex$ is possible 
    in $\Comp{n \log n}{\log n}$.
    \item Otherwise, if $Q$ is also self-join-free, then direct access for $Q$ by $\lex$ is not possible 
    in $\Comp{n \polylog n}{\polylog n}$,
    assuming \sparseBMM{} and \hyperclique{}.
\end{itemize}
\end{theorem}

\begin{theorem}\label{th:fds_lex_selection}
Let $Q$ be a CQ with unary FDs $\Delta$ and $\lex$ be a partial lexicographic order.
\begin{itemize}
    \item If $Q^+$ is free-connex, 
    then selection for $Q$ by $\lex$ is possible 
    in $\Comp{1}{n}$.
    \item Otherwise, if $Q$ is also self-join-free, then selection for $Q$ by $\lex$ is not possible 
    in $\Comp{1}{n \polylog n}$,
    assuming \seth{} and \hyperclique{}.
\end{itemize}
\end{theorem}

\subsection{A Note on General FDs}
\label{sec:generalFDs}

We discussed only unary FDs, where a single variable implies another.
The positive side of our results also holds for general FDs where a combination of variables may imply a variable. 
We simply need to take the general form of the extension (given an FD $x_1,\ldots,x_m\rightarrow y$, we add $y$ wherever all of $x_1,\ldots,x_m$ appear). If the extension has a tractable form, \Cref{thm:fds-known} and \Cref{thm:fds-known-lex} show that the original query is tractable too. 
However, extending the negative results requires a much more intricate analysis that goes beyond the scope of this work. 
Already for enumeration, even though Carmeli and Kr{\"{o}}ll~\cite{DBLP:journals/mst/CarmeliK20} showed a classificiation for general FDs when the extension is acyclic, 
the cyclic case is not resolved, and they provide a specific example of a CQ and FDs where the complexity is unknown.

\section{Conclusions}
\label{sec:conclusions}

We investigated the task of constructing a direct-access data structure to the output of a query with an ordering over the answers,
as well as the restriction of the problem to accessing a single answer (the selection problem).
We presented algorithms for fragments of the class of CQs for lexicographic and sum-of-weights orders. 
The direct access algorithms take quasilinear construction time in the size of the database, and logarithmic time for access.
For selection, our algorithms take quasilinear or even linear time.
We further showed that within the class of CQs without self-joins, our algorithms cover all the cases where these complexity guarantees are feasible, assuming conventional hypotheses in the theory of fine-grained complexity.
We were also able to precisely capture how the frontier of tractability changes under the presence of unary FDs.

This work opens up several directions for future work, including the generalization to more expressive queries (CQs with self-joins, union of CQs, negation, etc.), other kinds of orders (e.g., min/max over the tuple entries),
and a continuum of complexity guarantees
(beyond $\langle$quasilinear, logarithmic time$\rangle$). 

Generalizing the question posed at the beginning of the Introduction, we view this work as part of a bigger challenge that continues the line of research on \e{factorized representations} in databases~\cite{DBLP:conf/icdt/OlteanuZ12,DBLP:journals/sigmod/OlteanuS16}: how can we represent the output of a query in a way that, compared to the explicit representation, is fundamentally more compact and efficiently computable, yet equally useful to downstream operations?

\bibliographystyle{ACM-Reference-Format}
\bibliography{bibliography}

\clearpage
\appendix

\section{Nomenclature}

\begin{table}[h]
\centering
\small
\begin{tabularx}{\linewidth}{@{\hspace{0pt}} >{$}l<{$} @{\hspace{2mm}}X@{}} %
\hline
\textrm{Symbol}		& Definition 	\\
\hline
	R,S,T,U,R_1,R_2		& relation \\
	e, e_R, V, V_1, V_2 & atom/hyperedge/node of join tree \\
	\calS       & schema \\
	I           & database (instance) \\
    n           & size of $I$ (number of tuples) \\	
	\dom        & domain \\
	\mathbf{X}, \mathbf{Y}, \mathbf{Z}		& list of variables/attributes \\
	t			& tuple \\	
	x,y,z,u,v_1,v_2		& variable \\	
	Q			& CQ	\\
	q \in Q(I)  & query answer of CQ $Q$ over database $I$ \\
	Q(I)        & set of answers of $Q$ over $I$ \\
	\pi_\mathbf{X}(R)       & projection of on $\mathbf{X}$\\
	\mathbf{X}_f, \free(Q)  & free variables of query $Q$ \\
	\var(Q), \var(e)     & variables of query or atom \\
	\atoms(Q)   & set of query atoms \\
	\calH(Q) = (V, E)	& hypergraph associated with query $Q$ \\
	\calH_\free(Q) = (V, E)	& restriction of $\calH(Q)$ to free variables only \\
	T   & join tree \\
	\calV   & set of join tree nodes \\
	\freeind(Q)	& maximum number of independent free variables of $Q$ \\
    \mh(Q)	& number of maximal hyperedges (with respect to containment) in $\calH(Q)$ \\
    \mhfree(Q)	& number of free-maximal hyperedges = maximal hyperedges in $\calH_\free(Q)$ \\
    \lex = \angs{v_1, \ldots, v_m} & lexicographic order of variables \\
    w_x     & weight function for variable $x$: $\dom \rightarrow \R$\\
	w_Q    & weight function for query answers \\
	w       & short form for all $w_x$ and $w_Q$ \\
	\Sigma w  & sum-of-weights order \\
	\lambda  & a real-valued weight \\
	\preceq & total order over query answers \\
	\Pi & family of orders \\
	\text{order} & a binary relation as in partial/total order \\
	\text{ordering} & a sorted list according to an order \\
	R: \mathbf{X} \rightarrow \mathbf{Y} & FD where $\mathbf{X}$ implies $\mathbf{Y}$ in $R$ \\
	\Delta & set of FDs \\
	\Comp{n \log n}{\log n} & direct access with $\bigO(n \log n)$ preprocessing and $\bigO(\log n)$ per access\\
	\Comp{1}{n \log n} & selection in $\bigO(n \log n)$\\
\hline
\end{tabularx}
\end{table}

\end{document}